\documentclass[sigconf]{acmart}

\usepackage{flushend}
\usepackage[a-1b]{pdfx}   
\usepackage{balance}
\usepackage{pifont}

\usepackage{color}
\usepackage[labelfont=bf]{caption}
\usepackage{subcaption}
\usepackage{hyperref}
\usepackage{makecell}
\usepackage{fixltx2e}
\usepackage{rotating,multirow}
\usepackage{etoolbox,xspace}
\usepackage{algorithmicx}
\usepackage{algorithm}
\usepackage{fontawesome}
\usepackage{algpseudocode}
\usepackage{wrapfig}
\algtext*{EndWhile}
\algtext*{EndIf}
\algtext*{EndFunction}
\algtext*{EndFor}
\algrenewcommand\algorithmicrequire{\textbf{Input:}}
\algrenewcommand\algorithmicensure{\textbf{Output:}}

\usepackage{amsthm}

\usepackage{tabularx}
\usepackage{ragged2e}  
\usepackage{booktabs}  
\usepackage{textcomp}
\usepackage{url}
\usepackage{comment}
\usepackage{enumitem}
\usepackage[simplified]{pgf-umlcd}
\usepackage{tikz}

\newcolumntype{Y}{>{\RaggedRight\arraybackslash}X}

\newtheorem{definition}{Definition}
\newtheorem{theorem}{Theorem}
\newtheorem{lemma}{Lemma}

\newcommand{\eop}{\hspace*{\fill}\mbox{$\Box$}}     
\newenvironment{proof}{\paragraph{Proof:}}{\hfill$\square$}
\newcounter{example}
\renewcommand{\theexample}{\arabic{example}}
\newenvironment{example}{
        \vspace{1ex}
        \refstepcounter{example}
        {\noindent\bf Example \theexample:}}
	{\eop\vspace{1ex}}

\newcommand{\squishlist}{
 \begin{list}{$\bullet$}
  { \setlength{\itemsep}{0pt}
     \setlength{\parsep}{1pt}
     \setlength{\topsep}{1pt}
     \setlength{\partopsep}{0pt}
     \setlength{\leftmargin}{1em}
     \setlength{\labelwidth}{1em}
     \setlength{\labelsep}{0.5em} } }

\newcommand{\squishend}{
  \end{list}
}

\definecolor{americanrose}{rgb}{1.0, 0.01, 0.24}
\definecolor{airforceblue}{rgb}{0.36, 0.54, 0.66}
\definecolor{ao(english)}{rgb}{0.0, 0.5, 0.0}
\definecolor{ao}{rgb}{0.0, 0.0, 1.0}

\newcommand{\mohsen}[1]{\textcolor{red}{Mohsen: #1}}

\newcommand{\eat}[1]{}

\usepackage{xspace}

\newcommand{\Gee}{\mathcal{G}}
\newcommand{\gee}{\mathbf{g}}

\newcommand{\np}{{\sc NP}}

\newcommand{\eps}{\varepsilon}

\newcommand{\at}[1]{{\tt \small #1}\xspace}
\newcommand{\setc}{{\sc set cover}\xspace}
\newcommand{\fsc}{{\sc fair set cover}\xspace}
\newcommand{\kmc}{{max $k$-color cover}\xspace}
\newcommand{\wkmc}{{weighted max $k$-color cover}\xspace}

\newcommand{\Wratio}{\Delta}

\newcommand{\Groups}{\Gee}
\newcommand{\group}{\gee}
\renewcommand{\Re}{\mathbb{R}}%
\def\mparagraph#1{\par\noindent\textbf{#1.}\quad}

\newcommand{\poly}{\operatorname{poly}}

\newcommand{\NP}{\mathsf{NP}}
\newcommand{\element}{e}
\newcommand{\Elements}{U}
\newcommand{\Sets}{\mathcal{S}}
\newcommand{\set}{S}
\newcommand{\weight}{w}
\newcommand{\cover}{X}
\newcommand{\Wcover}{X_{\weight}}
\newcommand{\optCover}{\cover^*}
\newcommand{\optWCover}{\Wcover^*}
\newcommand{\fairprob}{FSC}
\newcommand{\fairWprob}{FWSC}
\newcommand{\Gfairprob}{GFSC}
\newcommand{\GfairWprob}{GFWSC}
\newcommand{\kcover}{C}
\newcommand{\newsize}{p}
\newcommand{\newsizeGroup}{p}
\newcommand{\fractionGroup}{f}
\newcommand{\uncovered}{\Elements^-}

\DeclareMathOperator*{\argmin}{arg\,min}
\allowdisplaybreaks

\copyrightyear{2025} 
\acmYear{2025} 
\setcopyright{acmlicensed}
\acmConference[KDD '25] {Proceedings of the 31st ACM SIGKDD Conference on Knowledge Discovery and Data Mining V.1}{August 3--7, 2025}{Toronto, ON, Canada.}
\acmBooktitle{Proceedings of the 31st ACM SIGKDD Conference on Knowledge Discovery and Data Mining V.1 (KDD '25), August 3--7, 2025, Toronto, ON, Canada}
\acmISBN{979-8-4007-1245-6/25/08} 
\acmDOI{10.1145/3690624.3709184}

\begin{document}
\title{Fair Set Cover}

\author{Mohsen Dehghankar}
\affiliation{%
  \institution{University of Illinois Chicago}
  \city{Chicago}
  \state{IL}
  \country{USA}
}
\email{mdehgh2@uic.edu}

\author{Rahul Raychaudhury}
\orcid{0009-0009-6610-4404}
\affiliation{%
  \institution{Duke University}
  \city{Durham}
  \state{NC}
  \country{USA}
}
\email{rahul.raychaudhury@duke.edu}

\author{Stavros Sintos}
\orcid{0000-0002-2114-8886}
\affiliation{%
  \institution{University of Illinois Chicago}
  \city{Chicago}
  \state{IL}
  \country{USA}
}
\email{stavros@uic.edu}

\author{Abolfazl Asudeh}
\orcid{0000-0002-5251-6186}
\affiliation{%
  \institution{University of Illinois Chicago}
  \city{Chicago}
  \state{IL}
  \country{USA}
}
\email{asudeh@uic.edu}

\renewcommand{\shortauthors}{Mohsen Dehghankar, Rahul Raychaudhury, Stavros Sintos, and Abolfazl Asudeh}

\begin{abstract}

  The potential harms of algorithmic decisions have ignited algorithmic fairness as a central topic in computer science.
  One of the fundamental problems in computer science is Set Cover, which has numerous applications with societal impacts, such as assembling a small team of individuals that collectively satisfy a range of expertise requirements.
  However, despite its broad application spectrum and significant potential impact, set cover has yet to be studied through the lens of fairness.

  Therefore, in this paper, we introduce \emph{Fair Set Cover}, which aims not only to cover with a minimum-size set but also to satisfy demographic parity in its selection of sets.
  To this end, we develop multiple versions of fair set cover, study their hardness, and devise efficient approximation algorithms for each variant.
  Notably, under certain assumptions, our algorithms always guarantee zero-unfairness, with only a small increase in the approximation ratio compared to regular set cover. 
  Furthermore, our experiments on various data sets and across different settings confirm the negligible price of fairness, as (a) the output size increases only slightly (if any) and (b) the time to compute the output does not significantly increase.
\end{abstract}







\maketitle

\vspace{-2mm}
\section{Introduction}
The abundance of data, coupled with AI and advanced algorithms, has revolutionized almost every aspect of human life. 
As data-driven technologies take root in our lives, their drawbacks and potential harms become increasingly evident. Subsequently, algorithmic fairness has become central in computer science research to minimize machine bias.
However, despite substantial focus on fairness in predictive analysis and machine learning~\cite{barocas2017fairness,mehrabi2021survey,pessach2022review}, 
traditional combinatorial optimization has received limited attention~\cite{wang2022balancing}. Nevertheless, combinatorial problems such as ranking~\cite{asudeh2019designing}, matching~\cite{esmaeili2023rawlsian,garcia2020fair}, resource allocation~\cite{donahue2020fairness,blanco2023fairness,jiang2021rawlsian}, and graph problems of influence maximization~\cite{tsang2019group} and edge recommendation~\cite{swift2022maximizing,bashardoust2023reducing} have massive societal impact.  

In this paper, we study the Set Cover problem through the lens of fairness.
As one of Karp's 21 \np-complete problems with various forms, set cover models a large set of real-world problems 
~\cite{vemuganti1998applications,revelle1976applications}, ranging from airline crew scheduling~\cite{rubin1973technique} and facility location (such as placing cell towers) to computational biology~\cite{cho2012chapter} and network security~\cite{ge2010application}, to name only a few.
Given a ground set $\Elements$ of elements and a family of subsets of $\Elements$, called $\Sets$, in the Set Cover problem, the final goal is to cover all the elements of $U$ by picking the least number of sets from $\Sets$.
Beyond its classical application examples, in the real world, set cover is frequently used to model problems that (either directly or indirectly) impact societies and human beings.
In such settings, it is vital to ensure that the set cover optimization (a) does not cause bias and discrimination against some (demographic) groups, and (b) promotes social values such as equity and diversity.
We further motivate the problem with a real-world application in Example~\ref{ex-1} for {\sc Team of Experts Formation}.

\begin{example}\label{ex-1}
    Consider an institution that wants to assemble a team of experts  (e.g., a data science company that wants to hire employees) that collectively satisfy a set of skills (e.g., \{\at{python}, \at{sql}, \at{data-visualization}, \at{statistics}, \at{deep-learning}, $\cdots$\}), i.e., for any skill, there must be at least one team member who has expertise in that skill. Naturally, the institution's goal is to cover all the skills with the smallest possible team. However, due to historical biases in society, solely optimizing for the team size may cause discrimination by mostly selecting team members from privileged groups. Unfortunately, many such {\em biased selections} have happened in the real world.
    For instance, investigating an incident where HP webcams failed to track \at{black} individuals~\cite{hp1}, revealed that HP engineers were mostly \at{white-male} and that caused the issue~\cite{hp2}.
    On the flip side, promoting social values such as diversity are always desirable in problems such as team formation (e.g., the data science company would like to ensure diversity in their hiring).
    In other words, not only do they want to satisfy all skills, but they also want to have an equal (or proportionate) representation of various demographic groups in the assembled team.
\end{example}
\vspace{-1mm}

Motivated by Example~\ref{ex-1} and other social applications of set cover, outlined in \S~\ref{sec:applications}, in this paper we introduce {\fsc}, where the objective is to find a smallest collection of sets to cover a universe of items, while ensuring (weighted) parity in the selection of sets from various groups. More specifically, given a set of groups $\Groups$, a set of elements $\Elements$, a family of sets $\Sets$ such that each set belongs to a group, and a coefficient $f_h$ for each group $\group_h\in \Groups$ with $\sum_{\group_h\in \Groups}f_h=1$, the goal is to select the minimum number of sets from $\Sets$ to cover all elements in $\Elements$, such that, for each group $\group_h\in \Groups$ the fraction of the number of sets from group $\group_h$ in the final solution is $f_h$. 
The formal definitions of our  problems are given in \S~\ref{sec:pre}.

In this paper, we revisit the set cover problem through the lens of group fairness. Our main contributions are as follows.

$\bullet$(\S~\ref{sec:applications}, \S~\ref{sec:pre}) 
    Demonstrating some of its applications, we propose the (unweighted and weighted) fair set cover problem, and study its complexity and hardness of approximation.
    We formulate fairness based on demographic parity in a general form that can be used to enforce count parity, as well as ratio parity, in selection among a set of binary or non-binary demographic groups.
    While various related problems have been studied in the literature (see \S~\ref{sec:related}), to the best of our knowledge, {\em we are the first} to formulate and study the problem of \fsc and its variations (such that the weighted case or the $\eps$-unfairness) over multiple groups.
   
    $\bullet$(\S~\ref{sec:FSC}, \S~\ref{sec:general}) 
    In the unweighted setting, we first propose a naive algorithm that achieves zero-unfairness but at a (large) cost of increasing the approximation ratio by a factor of the number of groups. Next, we propose a greedy algorithm that, while guaranteeing fairness, has the same approximation ratio as of the regular set cover; however, its running time is exponential in the number of groups.
    Therefore, we propose a faster (polynomial) algorithm that still guarantees fairness but slightly increases the approximation ratio by a fixed factor of $\frac{e}{e-1}$. Furthermore, we discuss a relaxed version of fairness constraints (called $\eps$-unfairness) and the trade-offs in the approximation ratios achieved by our algorithms.
   
    $\bullet$ (\S~\ref{sec:WFSC}, \S~\ref{sec:general})
    In the weighted setting, we follow the same sequence by first proposing a naive algorithm that guarantees fairness but at a cost of a major increase in the output size (approximation ratio). We then propose a greedy algorithm that improves the approximation ratio but runs in exponential time relative to the number of groups. Finally, we propose a faster algorithm but at the expense of a minor increase in the approximation ratio.

    $\bullet$ (\S~\ref{sec:experiments})
    To validate our theoretical findings, we conduct comprehensive experiments on real and synthetic datasets, across various settings, including binary vs. non-binary groups and fairness based on equal count vs. ratio parity. In summary, our experiments verify that our algorithms can achieve fairness with a negligible (if any) increase in the output size, and a small increase in the run time.


\section{Application Demonstration}\label{sec:applications}
In Example~\ref{ex-1}, we demonstrated an application of fair set cover for {\sc Team of Experts Formation}.
To demonstrate the application-span of our formulation, in the following, we briefly mention a few more applications.

\paragraph{\sc Equal Base-rate in representative data sets}
An equal base rate is defined as having an equal number of objects for different subgroups in the data set~\cite{kleinberg2016inherent}. Satisfying equal base-rate is important for satisfying multiple fairness metrics in down-stream machine learning task.
Consider the schema of a dataset of individual records (e.g., Chicago breast-cancer examination data), with a set of attributes $\mathbf{x}=\{x_1,\cdots,x_\eta\}$, where $\text{dom}(x_i)>1$ is the domain of $x_i$. In order to ensure representativeness in the data, suppose we would like to 
collect a minimal set that contains at least one instance from all possible level-$\ell$ value combinations.
A level-$\ell$ combination is a value-combination containing $\ell$ attributes in $\mathbf{x}$.
Each individual belongs to a demographic group (e.g., \at{White}, \at{Black}). 
In order to reduce potential unfairnesses for various groups (e.g., Chicago breast-cancer detection disparities~\cite{hirschman2007black}), it is required to satisfy equal base-rate in the dataset. This is an instance of \fsc, where each row (individual record) in the data set, covers all level-$\ell$ combinations it matches with.


\vspace{-2mm}
\paragraph{\sc Business license distribution}
Issuing a minimum number of business licences (e.g., cannabis dispensary) while covering the population of a city is a challenging task for the City Office.
Each business location covers regions within a certain (travel time) range.
Besides coverage, {\em equity} is a critical constraint in these settings.
For example, disparity against \at{Black}-owned cannabis dispensaries in Chicago is a well-known issue~\cite{chiagoDispensry}. Formulating and solving this problem as a fair set cover instance will resolve the disparity issues.


In Social Networks a similar application is {\sc influencer selection}, where companies would like to select a small set of ``innfluencers'' to reach out to their target application. Specifically, once an influencer (aka content creator) post a content about a product, their followers view it.
However, the company might also want to send their products to a diverse set of influencers, to prevent visual biases in their advertisement~\cite{adBiasInfluence}. Each influencer has a "cost" which is the amount of money that they should get to advertise a product. The company should minimize the advertising budget so the problem is mapped to the \emph{fair weighted set cover problem}.
\vspace{-2mm}
\paragraph{\sc Fair Clustering}
Clustering is a fundamental problem in Computer Science with many applications in as social network analysis, medical imaging, and anomaly detection~\cite{xu2005survey}. 
We are given a set of $n$ elements, where each element belongs to a group.
We are also given a clustering objective function that measures the clustering error.
The goal is to compute a set of $k$ centers to minimize the clustering error, ensuring that a specific fraction of centers is selected from each group.
Interestingly some of the known clustering problems can be modeled as a \setc instance (for example $k$-center clustering).
While fair clustering has been studied over different clustering objective functions, such as $k$-center, $k$-median, and $k$-means~\cite{hotegni2023approximation, jones2020fair, kleindessner2019fair}, there are still interesting open problems in the area. In the \emph{sum of radii clustering} the goal is to choose a set of $k$ balls (where the centers of the balls are elements in the input set) that cover all elements minimizing the sum of radii of the balls.
This problem is more useful than the $k$-center clustering because it reduces this dissection effect~\cite{hansen1997cluster, monma1989partitioning}.
The standard (unfair) version of this clustering problem has been thoroughly studied in the literature~\cite{doddi2000approximation, charikar2001clustering}.
Known constant approximation fixed parameter tractable algorithms are known for the fair sum of radii clustering~\cite{chen2024parameterized, inamdar2020capacitated}, however their running time depends exponentially on $k$.
The problem can be mapped to an instance of the \fsc problem with $n$ elements and $O(n^2)$ sets: The elements in \fsc are exactly the elements in the clustering problem. For every element $\element$ we construct $O(n)$ balls/sets with center $\element$ and radii all possible distances from $\element$ to any other element. Each set is assigned a weight which is equal to the radius of the ball. It is not hard to show that any solution of the weighted \fsc instance is a valid solution for the fair sum of radii clustering and vice versa. Hence, our new algorithms for the (weighted) \fsc problem can be used to derive an approximation algorithm for the fair sum of radii clustering in $O(\poly(n,k))$ time.
In the technical report~\cite{dehghankar2024fair}, we provide additional applications of the \fsc problem.




\section{Preliminaries}\label{sec:pre}
\mparagraph{Problem definition}
We are given a universe of $n$ elements $\Elements = \{\element_1, \element_2, \ldots, \element_n\}$, a family of $\mu$ sets $\Sets = \{\set_1, \set_2, \ldots, \set_\mu\}$, where $\cup_{i=1}^\mu\set_i=\Elements$.
We are also given a set of $k$ demographic groups $\Groups=\{\group_1,\ldots, \group_k\}$ (e.g., \{\at{white-male}, \at{black-female},\at{...}\}), while each set $\set_j$ in the family $\Sets$ belongs to a specific demographic group $\group(\set_j)\in \Groups$. In the rest of this paper, we interchangeably refer to the group of a set as its {\em color}. 
Let $\Sets_h\subseteq \Sets$ be the family of sets from $\Sets$ with color $\group_h\in \Groups$. Let $m_h$ be the number of sets of color $\group_h$ in $\Sets$, i.e., $m_h=|\Sets_h|$. 
A \emph{cover} $\cover\subseteq \Sets$ is a subset of $\Sets$ such that for any $\element\in \Elements$ there exists a set $\set\in \cover$ with $\element\in \set$. 
In other words, $X$ is a cover if $\bigcup_{s_i\in X}= \Elements$.
To simplify the notation, for a family of sets $\cover\subseteq \Sets$, we define $\cover\cap \Elements = (\bigcup_{\set\in\cover}\set)\cap \Elements$.

\vspace{-2.5mm}
\paragraph{Fairness definition} Following Example~\ref{ex-1} and the applications demonstrated in \S~\ref{sec:applications}, we adopt the group-fairness notion of {\em demographic parity}, as the (weighted) parity between the number of sets selected from each group. 
Specifically, given the non-negative coefficients $\{f_1,\cdots,f_k\}$ such that $\sum\limits_{h=1}^k f_h=1$, a cover $\cover$ is called a \emph{fair cover} if, for all groups $\gee_h\in \Gee$:
$
\big|\cover \cap \Sets_h\big| = f_h\, \big|\cover \big|$.

Two special cases based on this definition are the following:
\begin{itemize}[leftmargin=*]
    \item {\em Count-parity}: when $f_h=\frac{1}{k}$, $\forall \gee_h\in \Gee$, a fair set cover contains equal number from each group. 
    \item {\em Ratio-parity}: when $f_h=\frac{m_h}{\mu}$, $\forall \gee_h\in \Gee$, a fair set cover maintains the original ratios of groups in $\Sets$. 
\end{itemize}
Note that our fairness definition is not limited to the above cases, and includes any fractions (represented as rational numbers) specified by non-negative coefficients $f_h$ such that $\sum_{h=1}^k f_h=1$.




Next, we define the main problem studied in this paper.
Let $\weight: \Sets\rightarrow \Re^+$ be a weighted function over the sets in $\Sets$. By slightly abusing the notation, for a family of sets $C\subseteq \Sets$, we define $\weight(C)=\sum_{\set\in C}\weight(\set)$. The goal is to find a fair cover such that the selected sets in the cover has the smallest sum of weights.
\vspace{-1mm}
\begin{definition}[Generalized Fair Weighted Set Cover problem (\GfairWprob{})]
\label{prob:two}
    Given a universe of $n$ elements $\Elements$, a set of $k$ groups $\Groups=\{\group_1,\ldots, \group_k\}$, a fraction $\fractionGroup_h$ for each $\group_h\in \Groups$ such that their sum is equal to $1$, a family of $\mu$ subsets $\Sets$, such that each $\set\in \Sets$ is associated with a group $\group(\set)\in \Groups$, and a weighted function $\weight:\Sets\rightarrow \Re^+$ the goal is to compute a fair cover $\Wcover$ such that $\sum_{\set\in\Wcover}\weight(\set)$ is minimized.
\end{definition}
\vspace{-2mm}
If all weights are equal to $1$, i.e., $\weight(S)=1$, for every set $S\in \Sets$, then we call it the Generalized Fair Set Cover problem (\Gfairprob{}).
In the next sections, we denote the optimum fair cover with $\optCover$.

\vspace{2mm}
Under the {\em count-parity} fairness model, 
we call our problems, the Fair Set Cover (\fairprob{}) problem, and the Fair Weighted Set Cover (\fairWprob{}) problem, respectively.
Even though the \Gfairprob{} and \GfairWprob{} problems are more general than \fairprob{} and \fairWprob{}, respectively, the algorithmic techniques to solve \fairprob{} and \fairWprob{} can straightforwardly be extended to handle \Gfairprob{} and \GfairWprob{} with (almost) the same guarantees. In order to simplify the analysis, in the next sections we focus on \fairprob{} and \fairWprob{}. Then in \S~\ref{sec:general}, we describe how our algorithms can be extended to handle \Gfairprob{} and \GfairWprob{}.

\mparagraph{Hardness and Assumption} It is straightforward to show that \fairprob{} is \textsf{NP}-Complete. 
Recently, Barnab{\`o} et al.~\cite{barnabo2019algorithms} showed that \fairprob{} cannot be approximated with any sublinear approximation factor in polynomial time unless $\mathsf{P} = \mathsf{NP}$.
However, if there are \emph{enough} sets from each group, then in the next sections we design efficient algorithms with provable approximation guarantees.
More formally, for the \fairprob{} problem, we assume that every group contains at least $\Omega(\log n)\frac{|\optCover|}{k}$ sets.
For simplicity, in the next sections we also assume that every group has the same number of sets, i.e., for every group $\group_j$ it holds that $m_j=m$. Hence, $\mu=mk$. This only helps us to analyze our algorithms with respect to only three variables, $m, k, n$. All our algorithms can be extended to more general cases (where each color does not contain the same number of sets). We further highlight it in \S~\ref{sec:experiments}, where we run experiments on datasets with $m_i\neq m_j$ for $i\neq j$.


\mparagraph{$\eps$-unfairness} In some applications, the definition of fair cover might be too restrictive. Recall, that a cover $X$ is fair if there are exactly $f_h|X|$ elements from color $\group_h$ in cover $X$ (zero-unfairness).
Instead, we define the notion of \emph{$\eps$-unfair cover}. Given a parameter $\eps\in(0,1)$, a cover $X$ is called an $\eps$-unfair cover if, for each group $\group_h\in \Groups$, it holds that $(1-\eps)f_h|X|\leq |X\cap \Sets_h|\leq (1+\eps)f_h|X|$. In other words, instead of requesting exactly $f_h|X|$ elements from each group $\group_h$, we request any number of elements in the range $[(1-\eps)f_h|X|, (1+\eps)f_h|X|]$. Notice that any fair cover is also an $\eps$-unfair cover. All our defined problems can be studied in this setting. For example, we consider the \emph{$\eps$-\Gfairprob{} problem} that has the same objective with \Gfairprob{} problem, but the returned cover should be an $\eps$-unfair cover instead of a fair cover.
In \S~\ref{sec:epsilon}, we show how our new approximation algorithms for the \Gfairprob{} problem can be used to approximate the $\eps$-\Gfairprob{} problem.
\section{Fair Set Cover}\label{sec:FSC}

\subsection{Naive Algorithm}
We first describe a naive algorithm for the \fairprob{} problem with large approximation ratio. 
We execute the well-known greedy algorithm for the standard set cover problem (its pseudocode is shown in~the technical report~\cite{dehghankar2024fair}) and let $\cover$ be the set cover we get.
If $\cover$ does not satisfy the fairness requirements, we arbitrarily add the minimal number of sets from each color to equalize the number of sets from each color in our final cover. The pseudocode is shown in the technical report~\cite{dehghankar2024fair}.

\vspace{-1mm}
\begin{lemma}
    \label{thm:naiveAlg}
    There exists a $\big(k(\ln n +1)\big)$-approximation algorithm for the \fairprob{} problem that runs in $O(m\cdot k \cdot n)$ time. 
\end{lemma}

\vspace{-3mm}
\subsection{Greedy Algorithm}
\label{subsec:greedy}

\begin{figure}
    \centering
    \includegraphics[width=0.9\linewidth]{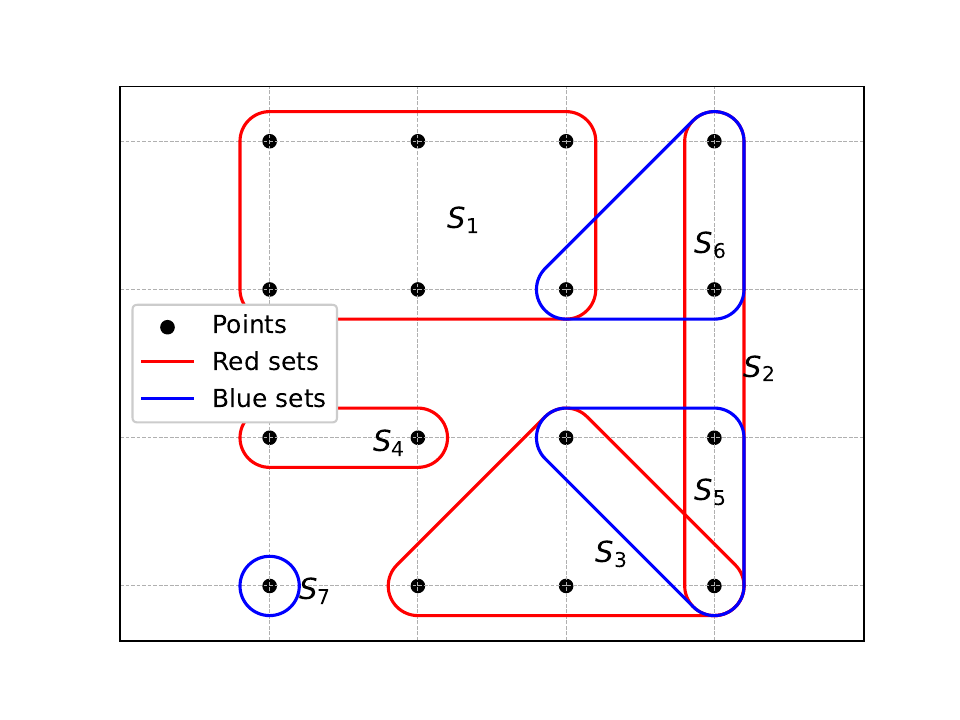}
    \vspace{-8mm}
    \caption{An illustration of a set system for the Fair Set Cover (FSC) problem, consisting of 16 points and 7 sets: 4 red sets and 3 blue sets. The Greedy Fair Set Cover algorithm (\S~\ref{subsec:greedy}) selects $(S_1, S_5)$ as the first pair of red and blue sets, followed by the pairs $(S_3, S_6)$ and $(S_4, S_1)$. The standard greedy algorithm selects the sets $S_1$, $S_2$, $S_3$, $S_4$, and $S_7$, resulting in 4 red sets and 1 blue set, highlighting its lack of fairness.}
    \label{fig:enter-label}
    \vspace{-3mm}
\end{figure}

In this subsection we propose an algorithm for the \fairprob{} problem whose approximation ratio {\em does not depend} on $k$.
For simplicity of explanation, we describe our algorithm for the case we have only two demographic groups, i.e., $\Groups=\{\group_1, \group_2\}$. In the end, our algorithm can be extended almost verbatim to $k>2$ groups.

Let $\cover=\emptyset$ be the cover we construct and let $\uncovered=\Elements$ be the set of uncovered elements.
We repeat the following steps until $\uncovered=\emptyset$.
In each iteration, we find the \emph{pair} of sets $(\set_A, \set_B)$ such that i) sets $\set_A, \set_B$ cover the maximum number of uncovered elements, i.e., $\uncovered\cap (\set_A \cup \set_B)$ is maximized, and ii) set $\set_A$ has color $\group_1$ and $\set_B$ has color $\group_2$, i.e., $\group(\set_A)=\group_1$ and $\group(\set_B)=\group_2$.
Once we find $\set_A, \set_B$ we update the set $\uncovered$ removing the new covered element, $\uncovered\gets \uncovered\setminus(\set_A\cup \set_B)$, and we add $\set_A, \set_B$ in $\cover$. In the end, after covering all elements ($\uncovered=\emptyset$), we return $\cover$. The pseudocode of our algorithm is shown in the technical report~\cite{dehghankar2024fair}.

\vspace{-1mm}
\paragraph{Analysis}
First, it is straightforward to see that $\cover$ is a fair cover. The algorithm stops when there is no uncovered element so $\cover$ is a cover. In every iteration, we add exactly one set of color $\group_1$ and one set of color $\group_2$, so $\cover$ is a fair cover.
\vspace{-1mm}
\begin{lemma}
    \label{lem:fairapprox}
    $|\cover|\leq (\ln n + 1)\cdot |\optCover|$.
\end{lemma}
\vspace{-3mm}
\begin{proof}
For each pair $\set_A, \set_B\in \Sets$, let $\set_{A,B}=\set_A\cup \set_B$.
We consider the following instance of the set cover problem.
We define the set $\Elements'=\Elements$ and $\Sets'=\{\set_{i,j}\mid \set_i\in \Sets_1, \set_j\in \Sets_2\}$. Let $C^*$ be the optimum solution of the set cover instance $(\Elements', \Sets')$. Any fair cover returned by our algorithm can be straightforwardly mapped to a valid set cover for $(\Elements', \Sets')$.
Hence, by definition, we have $|C^*|\leq \frac{1}{2} |\optCover|$.

Recall that the standard greedy algorithm for the set cover problem returns a $(\ln n +1)$-approximation.
We show that our algorithm implements such a greedy algorithm and returns a $(\ln n +1)$-approximation in the set cover instance $(\Elements', \Sets')$.
At the beginning of an iteration $i$, without loss of generality, assume that our algorithm has selected the pairs $(\set_{A_1}, \set_{B_1}),\ldots, (\set_{A_{i-1}}, \set_{B_{i-1}})$, so the sets 
$\set_{A_1,B_1}, \ldots, \set_{A_{i-1}, B_{i-1}}$ have been selected for the set cover instance $(\Elements', \Sets')$. Let $\uncovered$ be the set of uncovered elements in $\Elements'$ at the beginning of the $i$-th iteration.
Let $\set_{A_j,B_j}$ be the set in $\Sets'$ that covers the maximum number of elements in $\uncovered$. We consider two cases.

If $\set_{A_j}\notin \{\set_{A_1},\ldots, \set_{A_{i-1}}\}$ and $\set_{B_j}\notin \{\set_{B_1},\ldots, \set_{B_{i-1}}\}$, then by definition our greedy algorithm selects the pair $(\set_{A_j}, \set_{B_j})$ so the set $\set_{A_j,B_j}$ is added in the cover of the set cover instance $(\Elements', \Sets')$.

Next, consider the case where $\set_{A_j}\notin \{\set_{A_1},\ldots, \set_{A_{i-1}}\}$ and $\set_{B_j}\in \{\set_{B_1},\ldots, \set_{B_{i-1}}\}$ (the proof is equivalent if the set $S_{A_j}$ has been selected before). Since there are enough sets from each group, let $S_{B_h}$ be a set with color $\group_2$ such that $\set_{B_h}\notin \{\set_{B_1},\ldots, \set_{B_{i-1}}\}$. By definition, we have $|\set_{A_j,B_h}\cap \uncovered|\geq |\set_{A_j,B_j}\cap \uncovered|$. Our greedy algorithm selects the pair $\set_{A^*}, \set_{B^*}$ such that $|(\set_{A^*}\cup \set_{B^*})\cap \uncovered|$ is maximized, and $\set_{A^*}\notin \{\set_{A_1},\ldots, \set_{A_{i-1}}\}$, $\set_{B^*}\notin \{\set_{B_1},\ldots, \set_{B_{i-1}}\}$. Hence, $|\set_{A^*,B^*}\cap \uncovered|=|(\set_{A^*}\cup \set_{B^*})\cap \uncovered|\geq |(\set_{A_j}\cup \set_{B_h})\cap \uncovered|\geq |(\set_{A_j}\cup \set_{B_j})\cap \uncovered|=|\set_{A_j,B_j}\cap \uncovered|$.

In every iteration, our greedy algorithm computes a set $\set_{A,B}$ that covers the most uncovered elements in $\Elements'$.
Let $C$ be the family of sets returned by the greedy algorithm executed on $(\Elements', \Sets')$.
Every selected set in $C$ is a pair of sets in $\Sets$.
Hence,
$|\cover|\leq 2(\ln n +1)|C^*|$.
Furthermore,
$|C^*|\leq \frac{1}{2} |\optCover|$ as shown in the first paragraph of the proof. We conclude that $|\cover|\leq (\ln n + 1)|\optCover|$.
\end{proof}

Notice that there exists $O(m_1m_2)=O(m^2)$  pairs of sets with different colors. At the beginning of the algorithm, for every pair of sets we compute the number of elements they cover in $O(m^2n)$ time. Each time we cover a new element we update the counters in $O(m^2)$ time. In total, our algorithm runs in $O(m^2n)$ time.
This algorithm can be extended to $k$ groups, straightforwardly. In each iteration, we find the $k$ sets of different color that cover the the most uncovered elements. The running time increases to $O(m^kn)$.

\vspace{-2mm}
\begin{theorem}
    \label{thm:fairCover}
    There exists an $(\ln n +1)$-approximation algorithm for the \fairprob{} problem that runs in $O(m^kn)$ time.
\end{theorem}

While our algorithm returns an $O(\log n)$-approximation, which is the best that someone can hope for the \fairprob{}, the running time depends {\em exponentially} on the number of groups $k$. Next, we propose two faster $O(\poly(n,m,k))$ time randomized algorithms with the same asymptotic approximation.
\subsection{Faster algorithms}
\label{subsec:fasterAlg}

The expensive part of the greedy algorithm is the selection of $k$ sets of different color that cover the most uncovered elements, in every iteration.
We focus on one iteration of the greedy algorithm. 
Intuitively,
we need to find $k$ sets, each from a different color, that maximizes the coverage over the uncovered elements.
We call this problem, \kmc. Formally,


\begin{definition}[\kmc] 
    Consider a set $\Elements^-\subseteq \Elements$, called the set of uncovered elements, 
    a family of sets $\Sets^-\subseteq \Sets$, called the sets not selected so far, and
    the value $k$ showing the number colors (groups) that sets belong to.
    Our goal is to select $k$ sets $X\subseteq \Sets^-$, one from each color, such that $\big| X\cap \Elements^-\big|$ is maximized.
\end{definition}


The \kmc problem is $\NP$-complete. 
The proof is straight-forward using the reduction from the traditional max $k$-cover (MC) problem by creating $k$ copies of each original set, while assigning each copy to a unique color. Now the max-cover problem has a solution with coverage $\ell$, if and only if the $k$-color MC has a solution with coverage $\ell$.
Since \kmc is $\NP$-complete, we propose constant-approximation polynomial algorithms for it.

\paragraph{\textbf{$\frac{1}{2}$- approximation algorithm}}
Inspired by~\cite{ma2021fair}, where the authors also study a version of the \kmc problem (to design efficient algorithms for the \emph{fair regret minimizing set} problem), we propose the following algorithm.
We run the standard greedy algorithm for the max $k$-cover problem. Let $Y=\Elements^-$ be the current set of uncovered elements in $\Elements^-$ and let $\kcover=\emptyset$ be the family of $k$ sets we return. The next steps are repeated until $|\kcover|=k$ or $Y=\emptyset$. In each iteration, the algorithm chooses the set $\set\in\Sets^-$ such that $|\set\cap Y|$ is maximized. If $\kcover$ contains another set with color $\group(\set)$, then we remove $\set$ and we continue with the next iteration. Otherwise, we add $\set$ in $\kcover$ and we update $Y=Y\setminus \set$.
If $Y=\emptyset$ and $\kcover$ does not contain sets from all groups, we add in $\kcover$ exactly one set per missing group.

\vspace{-2mm}
\begin{theorem}
\label{thm:SimplefairMaxKCover}
    There exists an $\frac{1}{2}$-approximation algorithm for the \kmc problem that runs in $O(n\mu)$ time.
\end{theorem}
\vspace{-1em}
\begin{proof}
    By definition, $\kcover$ contains one set from each group in $\Groups$, so it is a valid $k$-color cover.
Let $\set_h$ be the $h$-th set added in $\kcover$.
    Without loss of generality, assume $\group(\set_h)=\group_h$, for $h=1,\ldots, k$. Let $O^*=\{O_1,\ldots, O_k\}\subseteq \Sets^-$ be the optimum solution for the \kmc problem. Without loss of generality, assume that $\group(O_h)=\group_h$, for $h=1,\ldots, k$. We have,
    \vspace{-1.5mm}
    \begin{align*}
        |\kcover|&=\sum_{i=1}^k\left|\set_i\setminus (\cup_{h=1..i-1}\set_h\right)|\geq \sum_{i=1}^k\left|O_i\setminus \cup_{h=1}^{i-1}\set_h\right |\geq \sum_{i=1}^k\left|O_i\setminus \kcover\right |\\
        &\geq |O^*|-|\kcover|.
    \end{align*}
    Hence $|\kcover|\geq \frac{1}{2}|O^*|$. Each time we cover a new element, we update the number of uncovered elements in each set in $O(\mu)$ time. Hence, the algorithm runs in $O(n\mu)$ time.
\end{proof}
\paragraph{\textbf{$(1-\frac{1}{e})$-approximation algorithm}}
We design a randomized \textsf{LP}-based algorithm for the \kmc problem with better approximation ratio.
For every element $\element_j\in \Elements^-$ we define the variable $y_j$. For every set $\set_i\in \Sets^-$ we define the variable $x_i$. Let $\Sets_h^-=\{\set\in\Sets^-\mid \set\in\Sets_h\}$.
The next Integer Program represents an instance of the \kmc problem. 
\begin{align*}
\label{fast_lp_1}
\max \sum_j y_j&\\
\sum_{i: \set_i\in\Sets^-_h}x_i&=1,\quad \forall \group_h\in \Groups\\
\sum_{i: \element_j\in \set_i} x_i&\geq y_j, \quad \forall \element_j\in \Elements^-\\
x_i\in \{0,1\}, \forall \set_i\in\Sets^-,
\quad&y_j\in \{0,1\},\forall \element_j\in \Elements^-
\end{align*}
We then replace the integer variables to continuous ($x_i\in [0,1]$ and $y_j\in [0,1]$) and use a polynomial time \textsf{LP} solver to compute the solution of the \textsf{LP} relaxation.
Let $x_i^*$, $y_j^*$ be the values of the variables in the optimum solution.
For every group $\group_h\in \Groups$, we sample exactly one set from $\Sets^-_h$, using the probabilities $\{x_i^*\mid \set_i\in \Sets^-_h\}$.
Let $\kcover$ be the family of $k$ sets we return.

\paragraph{\textbf{Analysis}}

By definition $\kcover$ contains exactly one set from every group $\group_i$. Next, we show the approximation factor.
\begin{lemma}
\label{lem:lem1}
An element $\element_j$ is covered by a set in $\kcover$ with probability at least $(1-\frac{1}{e})y_j^*$.
\end{lemma}
\begin{proof}
By definition, $\element_j$ is not covered by a set in $\Sets^-_h$ with probability

\vspace{-6mm}
$$1-\sum_{i: \element_j\in \set_i, \set_i\in\Sets^-_h}x_i^*\leq e^{-\sum_{i: \element_j\in \set_i, \set_i\in\Sets^-_h}x_i^*}.$$

\vspace{-1mm}
The inequality holds because of the well know inequality $1+x\leq e^x$ for every $x\in \Re$.
Hence,
\begin{align*}
Pr[\element_j \notin \kcover]&=\prod_{\group_h\in\Groups}(1-\sum_{i: \element_j\in \set_i, \set_i\in\Sets^-_h}x_i^*)
\leq \prod_{\group_h\in\Groups}e^{-\sum_{i: \element_j\in \set_i, \set_i\in\Sets^-_h}x_i^*}
\\&=e^{-\sum_{\group_h\in\Groups}\sum_{i: \element_j\in \set_i, \set_i\in\Sets^-_h}x_i^*}
=e^{-\sum_{i: \element_j\in \set_i}x_i^*} 
\leq e^{-y_j^*}.
\end{align*}
We conclude that
$Pr[\element_j \in \kcover]\geq 1-e^{-y_j^*}\geq (1-\frac{1}{e})y_j^*$.
\end{proof}


\vspace{-1mm}
\begin{lemma}
\label{lem:lem2}
The expected number of elements covered by sets in $\kcover$ is at least $(1-\frac{1}{e})\mathsf{OPT}$, where $\mathsf{OPT}$
is the number of covered elements in the optimum solution for the \kmc problem.
\end{lemma}
\begin{proof}
Let $\mathsf{OPT}_{LP}=\sum_j y_j^*$ be the optimum solution of the LP above.
Let $Z_j$ be a random variable which is $1$ if $\element_j$ is covered in $\kcover$, and $0$ otherwise.
We have
\begin{align*}
    E\Big[\sum_{j} Z_j\Big]&=\sum_j E[Z_j]
    =\sum_j Pr[Z_j=1]
    =\sum_j Pr[\element_j \in\kcover]\\&
    \geq \sum_j (1-\frac{1}{e})y_j^*
    =(1-\frac{1}{e})\mathsf{OPT}_{LP}
    \geq (1-\frac{1}{e})\mathsf{OPT}.
    \vspace{-2mm}
\end{align*}
\vspace{-1mm}
\end{proof}

Let $\mathcal{L}(a,b)$ be the running time to solve an \textsf{LP} with $a$ constraints and $b$ variables. Our algorithm for the \kmc problem runs in $O(\mathcal{L}(n+k,2n+2\mu))$ time.

Putting everything together, we conclude with the next theorem.
\begin{theorem}
    \label{thm:fairMaxKCover}
    There exists a randomized $(1-\frac{1}{e})$-approximation algorithm for the \kmc problem that runs in $O(\mathcal{L}(n+k,2n+2\mu))$ time. The approximation factor holds in expectation.
\end{theorem}

\vspace{-2mm}
\paragraph{\textbf{Faster algorithm for \fairprob{}}}
We combine the results from Theorem~\ref{thm:fairCover} and Theorem~\ref{thm:fairMaxKCover} (or Theorem~\ref{thm:SimplefairMaxKCover}) to get a faster algorithm for the \fairprob{} problem.
In each iteration of the greedy algorithm from Theorem~\ref{thm:fairCover} we execute the algorithm from Theorem~\ref{thm:fairMaxKCover} (or Theorem~\ref{thm:SimplefairMaxKCover}).
The pseudocode of the overall method is shown in the technical report~\cite{dehghankar2024fair}.

It is known~\cite{young2008greedy} that the greedy algorithm for the standard Set Cover problem, where in each iteration it selects a set that covers a $\beta$-approximation (for $\beta<1$) of the maximum number of uncovered elements returns a $\frac{1}{\beta}(\ln n +1)$-approximation. From the proof of Lemma~\ref{lem:fairapprox} we mapped the \fairprob{} problem to an instance of the standard set cover problem, so the approximation ratio of our new algorithm is $\frac{e}{e-1}(\ln n +1)$ (or $2(\ln n +1)$).

The algorithm from Theorem~\ref{thm:fairMaxKCover} runs in $O(\mathcal{L}(n+k,2n+2m\cdot k))$. We execute this algorithm in every iteration of the greedy algorithm. In the worst case the greedy algorithm selects all sets in $\Sets$, so a loose upper bound on the number of iterations is $O(m)$.
The number of iterations can also be bounded as follows. If $I_{\mathsf{Greedy}}$ is the number of iterations of the greedy algorithm, then by definition it holds that $k\cdot I_{\mathsf{Greedy}}\leq \frac{e}{e-1}(\ln n +1)|\optCover|\Leftrightarrow I_{\mathsf{Greedy}}=O(\frac{|\optCover|}{k}\log n)$.

\begin{theorem}
    \label{thm:fasterAlg}
    There exists a randomized $\frac{e}{e-1}(\ln n +1)$-approximation algorithm for the \fairprob{} problem that runs in
    $O(\frac{|\optCover|}{k}\mathcal{L}(n+k,2n+2m\cdot k)\log n)$ time. The approximation factor holds in expectation. Furthermore, there exists a $2(\ln n +1)$-approximation algorithm that runs in
    $O(|\optCover|\cdot n\cdot m\cdot\log n)$ time.
\end{theorem}

\section{Weighted Fair Set Cover}
\label{sec:WFSC}
Throughout this section we define $\Wratio = \frac{\weight_{\mathsf{max}}}{\weight_{\mathsf{min}}}$, where $\weight_{\mathsf{max}} = \max\{\weight(\set_i) | S_i \in \Sets\}$ and $\weight_{\mathsf{min}} = \min \{\weight(\set_i) | \set_i \in \Sets\}$.
In this section we need the assumption that for every group $\group_h\in\Groups$ there are at least $\Omega(\Delta\log n)\cdot \frac{|\optCover|}{k}$ sets.
In Appendix~\ref{appndx:sec5}, we show that a naive algorithm for the \fairWprob{}, similar to the naive algorithm we proposed for the \fairprob{} returns a $k\Wratio(\ln n +1)$-approximation algorithm that runs in $O(m\cdot k\cdot n)$ time.
\begin{theorem} \label{thm:NaivefairWCover}
    There exists a $k\Wratio(\ln n +1)$-approximation algorithm for the \fairWprob{} problem that runs in $O(m\cdot k\cdot n)$ time.
\end{theorem}

\vspace{-4mm}
\subsection{Greedy Algorithm}
\vspace{-1mm}
For simplicity, we describe our algorithm for the case we have only two demographic groups, i.e., $\Groups=\{\group_1, \group_2\}$. In the end, our algorithm can be extended almost verbatim to $k>2$ groups.
Let $\Wcover=\emptyset$ be the cover we construct and let $\uncovered=\Elements$ be the set of uncovered elements.
We repeat the following steps until $\uncovered=\emptyset$.
In each iteration, we find the pair of sets $(\set_A, \set_B)$ such that i) the ratio $\frac{\weight(\set_A)+\weight(\set_B)}{|(\set_A\cup \set_B)\cap \uncovered|}$ is minimized, and ii) $\group(\set_A)=\group_1$ and $\group(\set_B)=\group_2$.
Once we find $\set_A, \set_B$ we update the set $\uncovered$ removing the new covered element, $\uncovered\gets \uncovered\setminus(\set_A\cup \set_B)$, and we add $\set_A, \set_B$ in $\Wcover$. In the end, after covering all elements we return $\Wcover$.

The proof of the next Theorem can be found in Appendix~\ref{appndx:sec5}.

\vspace{-2mm}
\begin{theorem}
    \label{thm:fairWCover}
    There exists a $\Wratio(\ln n + 1)$-approximation algorithm for the \fairWprob{} problem that runs in $O(m^kn)$ time.
\end{theorem}

\vspace{-4mm}
\subsection{Faster algorithm}


The expensive part of the previous greedy algorithm is the selection of $k$ sets of different color to minimize the ratio of the weight over the number of new covered elements, in every iteration.
We focus on one iteration of the greedy algorithm, solving an instance of the \emph{\wkmc} problem.
We are given a set of $n$ elements $\Elements^-\subseteq\Elements$, a family of $\mu$ sets $\Sets^-\subseteq \Sets$, a weight function $\weight:\Sets^-\rightarrow \Re^+$, and a set of $k$ groups $\Groups$ such that every set belongs to a group. The goal is to select a family of exactly $k$ sets $\Wcover\subseteq \Sets^-$ such that i) $\frac{\weight(\Wcover)}{|\Wcover\cap \Elements^-|}$ is minimized, and ii) there exists exactly one set from every color $\group_h\in \Groups$.
An efficient constant approximation algorithm for the \wkmc problem could be used in every iteration of the greedy algorithm to get an $O(\Wratio \log n)$-approximation in $O(\poly(n,m,k))$ time.

\vspace{-1mm}
\newcommand{\threshold}{\tau}
\paragraph{\textbf{Algorithm}}
We design a randomized \textsf{LP}-based algorithm for the \wkmc problem. We note that our new algorithm has significant changes compared to the \textsf{LP}-based algorithm we proposed for the unweighted case. The reason is that the objective of finding a fair cover $\Wcover$ that minimizes $\frac{\weight(\Wcover)}{|\Wcover\cap \Elements^-|}$ can not be represented as a linear constraint. Instead, we iteratively solve a constrained version of the \wkmc problem:
For a parameter $\tau\in [1,n]$, the goal is to select a family of exactly $k$ sets $\Wcover\subset \Sets^-$ such that i) $\weight(\Wcover)$ is minimized, ii) $|\Wcover\cap \Elements^-|\geq \tau$, and iii) there exists exactly one set from every color $\group_h\in \Groups$ in $\Wcover$.
The main idea of our algorithm is that for each $\tau=1,\ldots, n$, we compute efficiently an approximate solution for the constrained \wkmc problem. In the end, we return the best solution over all $\tau$'s.

For every $\tau=1,\ldots, n$, we solve an instance of the constrained \wkmc problem.
The next Integer Program represents an instance of the constrained \wkmc problem with respect to $\tau$.

\vspace{-4mm}
\begin{align*}
\min \sum_{\set_i\in\Sets^-} \weight(\set_i)&x_i\\
\sum_{i: \set_i\in\Sets^-_h}x_i&=1,\quad \forall \group_h\in \Groups\\
\sum_{i: \element_j\in \set_i} x_i&\geq y_j, \quad \forall \element_j\in \Elements^-\\
\sum_{\element_j\in \Elements^-} y_j&\geq \tau\\
x_i\in \{0,1\},\forall \set_i\in\Sets^-,\quad &
y_j\in \{0,1\},\forall \element_j\in \Elements^-
\end{align*}
We then replace the integer variables to continuous ($x_i\in [0,1]$ and $y_j\in [0,1]$), and use a polynomial time \textsf{LP} solver to compute the solution of the \textsf{LP} relaxation.


Let $x_i^{(\tau)}$, $y_j^{(\tau)}$ be the values of the variables in the optimum solution.
For every group $\group_h\in \Groups$, we sample exactly one set from $\Sets^-_h$, using the probabilities $x_i^{(\tau)}$.
Let $\kcover^{(\tau)}$ be the sampled family of $k$ sets. If $|\kcover^{(\tau)}\cap \Elements^-|<\frac{1}{2}(1-\frac{1}{e})\tau$, then we skip the sampled family $\kcover^{(\tau)}$ and we re-sample from scratch. In the end, let $C^{(\tau)}$ be a family of $k$ sets returned by the sampling procedure such that $|\kcover^{(\tau)}\cap \Elements^-|\geq \frac{1}{2}(1-\frac{1}{e})\tau$.
Finally, we compute $r^{(\tau)}=\frac{\sum_{\set_i\in \Sets^-}\weight(\set_i)x_i^{(\tau)}}{|\kcover^{(\tau)}\cap \Elements^-|}$.
We repeat the algorithm for every $\tau=1,\ldots, n$, and we set
$\tau^*=\argmin_{\tau=1,\ldots, n}r^{(\tau)}$.
We return the family of $k$ sets $\kcover^{(\tau^*)}$.

We show the proof of the next theorem in Appendix~\ref{appndx:sec5}.
\vspace{-1mm}
\begin{theorem}    \label{thm:WeightfairMaxKCover}
    There exists a randomized $\frac{e}{e-1}$-approximation algorithm for the \wkmc problem that runs in $O(n\mathcal{L}(n+k,2n+2\mu+1)+kn^3)$ time.
    The approximation factor and the running time hold in expectation.
\end{theorem}
\vspace{-3mm}
\paragraph{\textbf{Faster algorithm for \fairWprob{}}}
We combine the results from Theorem~\ref{thm:fairWCover} and Theorem~\ref{thm:WeightfairMaxKCover} and we get the following theorem.
\vspace{-1mm}
\begin{theorem}
    \label{thm:fasterAlgW}
    There exists a randomized $\frac{e}{e-1}\Wratio(\ln n + 1)$-approximation algorithm for the \fairWprob{} problem that runs in $O(mn\cdot \mathcal{L}(n+k,2n+2m\cdot k+1)+mkn^3)$ time. The approximation factor and the running time hold in expectation.
\end{theorem}

\vspace{-4mm}
\section{Generalized Fair Set Cover}\label{sec:general}

Due to space limitations, we defer the details of our algorithms for \Gfairprob{} and \GfairWprob{} problems to Appendix~\ref{appndx:Sec6}.
Recall that in these problems, each color $\group_h$ is associated with a fraction $\fractionGroup_h\in [0,1]$, such that $\sum_{\group_h\in \Groups}\fractionGroup_h=1$, and the goal is to return a (fair) cover $\cover$ such that $\frac{|\cover\cap \Sets_h|}{|\cover|}=\fractionGroup_h$, for every color $\group_h\in \Groups$.
Our greedy (and faster greedy) algorithms from the previous sections (for the \fairprob{} and \fairWprob{} problems) easily extend to \Gfairprob{} and \GfairWprob{} with similar theoretical guarantees. The high level idea is the following: In each iteration of the greedy algorithm, instead of selecting exactly one set from each group, the algorithm selects $\newsizeGroup_h$ sets for group $\group_h$. The selection is made such that $\fractionGroup_h=\frac{p_h}{\sum_{\group_j\in\Groups} p_j}$ is satisfied for each group $\group_h$. As the greedy algorithm operates for integral rounds, the final set cover also satisfies the ratio constraints.

\vspace{-2mm}
\subsection{The $\eps$-\Gfairprob{} problem }\label{sec:epsilon}
Our algorithms for the \Gfairprob{} problem can be used
for the less constrained $\eps$-\Gfairprob{} problem, where, instead of zero unfairness, the goal is to ensure a maximum unfairness of $\eps$. Particularly, 
in Theorem~\ref{thm:epsUnfair}, we show that a solver for the \Gfairprob{} problem provides a good approximation for the $\eps$-\Gfairprob{} problem.

\begin{theorem}
\label{thm:epsUnfair}
    Any $\alpha$-approximation algorithm for the \Gfairprob{} problem is  an $\left((1+\eps)\alpha\right)$-approximation for the $\eps$-\Gfairprob{} problem.
\end{theorem}
\begin{proof}
    Let $\optCover$ be the optimum solution to the \Gfairprob{} problem and let  $\optCover_{\eps}$ be the optimum solution to the $\eps$-\Gfairprob{} problem. 
Starting from $\optCover_\eps$ we find a feasible cover $X_\eps'$ for the \Gfairprob{} problem.
Initially, $X_\eps'=\optCover_\eps$.
Then, we add
$(1+\eps)f_h|\optCover_\eps|-|\optCover_\eps\cap \Sets_h|$ sets from group $\group_h$, in $\optCover_\eps$.
In the end, notice that $|X_\eps'\cap \Sets_h|=(1+\eps)f_h|\optCover_\eps|$ for every group $\group_h$, and $|X_\eps'|=(1+\eps)|\optCover_\eps|$.
Hence, $\frac{|X_\eps'\cap \Sets_h|}{|X_\eps'|}=f_h$, so $X_\eps'$ is a solution for the \Gfairprob{} problem and $|\optCover|\leq |X_\eps'|$, by definition. We conclude that $|\optCover|\leq |X_\eps'|= (1+\eps)|\optCover_\eps|$.
 Let $\mathcal{A}$ be an $\alpha$-approximation algorithm for the \Gfairprob{} and $X$ be the output of $\mathcal{A}$. We have,
    $|X|\leq \alpha |\optCover|\leq (1+\eps)\alpha|\optCover_\eps|$.
\end{proof}

As a result, all the approximation algorithms proposed in the previous sections also approximate the $\eps$-\Gfairprob{} problem with asymptotically the same approximation ratio of $O(\ln n)$.

\vspace{-2mm}
\section{Experiments}
\label{sec:experiments}
\begin{figure*}[!tb]
\centering
    \begin{subfigure}[t]{0.32\linewidth}
        \includegraphics[width=.95\linewidth]{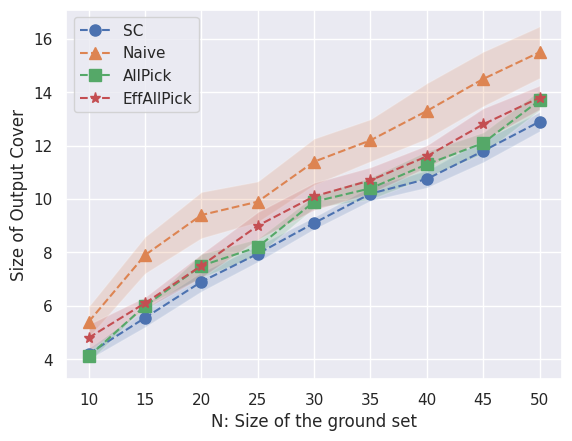}
        \caption{Output size}
        \label{fig:cover_size_resume}
    \end{subfigure}
    \hfill
    \begin{subfigure}[t]{0.32\linewidth}
        \includegraphics[width=.95\linewidth]{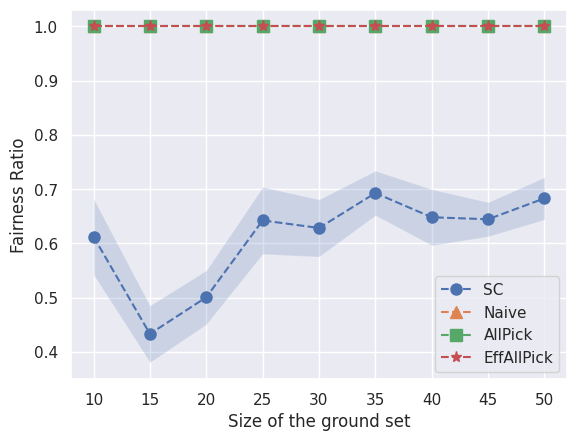}
        \caption{Fairness}
        \label{fig:fairness_resume}
    \end{subfigure}
    \hfill
    \begin{subfigure}[t]{0.32\linewidth}
        \includegraphics[width=.95\linewidth]{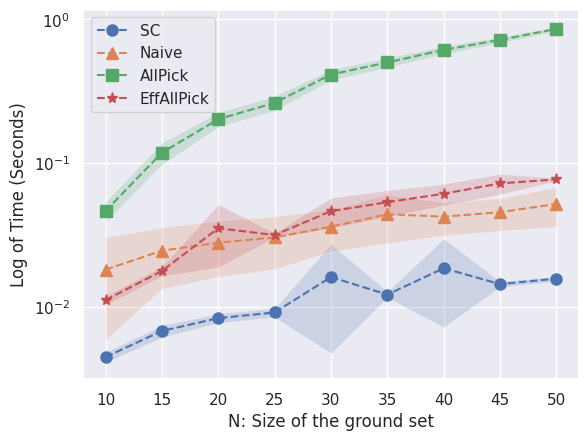}
        \caption{Time}
        \label{fig:running_time_resume}
    \end{subfigure}
\vspace{-4mm}
\caption{Resume Skills: Comparing (a) the size of output cover, (b) fairness (1 is the max fairness), and (c) running time for various ground set sizes $N$. The results are averaged over 20 samples for each $N$.}\label{fig:resume}
\vspace{-3mm}
\end{figure*}

\begin{figure*}
\centering
    \begin{subfigure}[t]{0.32\linewidth}
        \includegraphics[width=.95\linewidth]{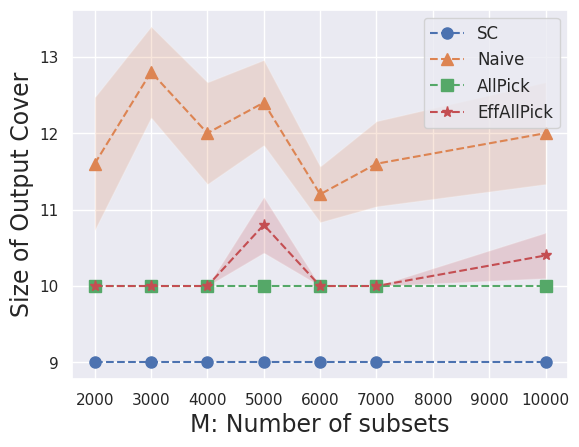}
        \caption{Output size}
    \end{subfigure}
    \hfill
    \begin{subfigure}[t]{0.32\linewidth}
        \includegraphics[width=.95\linewidth]{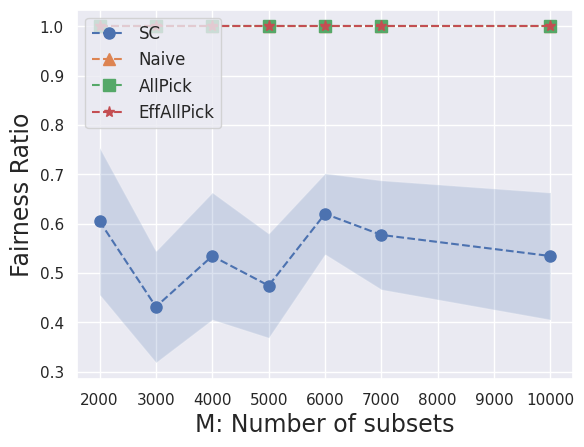}
        \caption{Fairness}
    \end{subfigure}
    \hfill
    \begin{subfigure}[t]{0.32\linewidth}
        \includegraphics[width=.95\linewidth]{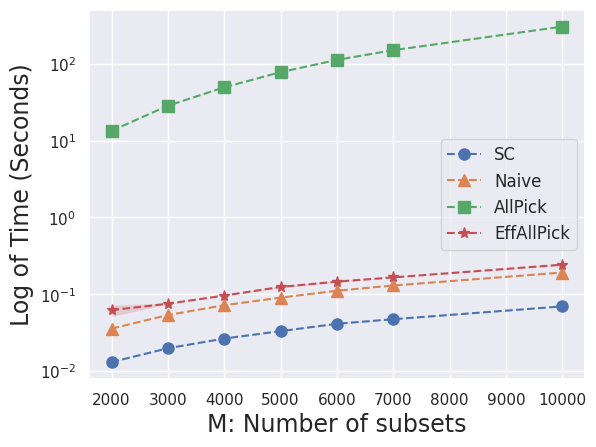}
        \caption{Time}
    \end{subfigure}
\vspace{-4mm}
\caption{Comparing (a) the size of output cover, (b) fairness based on ratio-parity, and (c) running time (averaged over 20 samples) on Adult dataset. The x-axis shows different values of $M$ (size of family of subsets or rows in the sampled dataset).}
\label{fig:exp:adult}
\end{figure*}

\vspace{-1mm}
Having theoretically analyzed the proposed algorithms, in this section we experimentally evaluate the efficiency and efficacy of our algorithms on real and synthetic datasets.

{\em Datasets:} We use 3 real datasets, one semi-synthesized, and one synthesized dataset to evaluate our methods. Our first set of experiments (\S~\ref{sec:exp:validation} and \S~\ref{sec:exp:resume}) are motivated by Example~\ref{ex-1} (Team of Experts Formation), for which we use the {\bf Strategeion Resume Skills} dataset\footnote{\url{http://bit.ly/2SeS4xo}}. It consists of 1986 rows as candidates with different skills. The skill of a candidate is shown as 218 binary columns each corresponding to one skill. As an AI Ethics benchmark, the dataset also contains demographic group information, including the (binary) gender. The male-to-female ratio in this dataset is $0.8$. 

Motivated by the "Fair survey data collection" application demonstrated in \S~\ref{sec:applications}, we use two additional real datasets: {\bf Adult} Income and {\bf COMPAS} (\S~\ref{sec:exp:real}). The {\bf Adult} Income Dataset ~\cite{Dua2019} consists of demographic information for individuals, which was extracted from the 1994 Census database. The dataset is primarily used to predict whether a person earns more than \$50,000 per year based on various attributes (14 columns and 48,842 rows). Gender is the sensitive attribute in this dataset. The {\bf COMPAS} dataset ~\cite{angwin2016compas} contains data used to evaluate the effectiveness and fairness of the COMPAS recidivism prediction algorithm. This dataset includes criminal history, jail and prison time, demographics, and COMPAS risk scores for defendants from Broward County, Florida (52 columns and 18,316 rows), while race is the sensitive attribute.

The third set of experiments (\S~\ref{sec:exp:popsim}) are using the {\bf POPSIM} ~\cite{nguyen2023popsim} dataset. {\bf POPSIM} is a semi-synthetic dataset that combines population statistics along with a geo-database. It is used to represent individual-level data with demographic information for the state of Illinois. Race is used as the sensitive attribute ($5$ groups). It contains 2,000,000 locations represented as points in $\Re^2$.

For the last set of experiments, (\S~\ref{sec:exp:2}), we synthesized a dataset consisting of a number of points as ground set and a family of subsets with different coverage distributions on this ground set.
More details about this dataset is provided in the technical report~\cite{dehghankar2024fair}.

{\em Algorithms:} We compare four algorithms over all the experiments: Standard Greedy Set Cover ({\bf SC}), Naive Algorithm ({\bf Naive}), Fair Greedy Algorithm ({\bf AllPick}), Efficient or Faster Fair Greedy Algorithm ({\bf EffAllPick}).
We also compare these algorithms with the optimum solution (using the brute-force algorithm for Standard Set Cover ({\bf Opt-SC}) and a brute-force algorithm for Fair Set Cover ({\bf Opt-Fair})). These algorithms simply check all possible covers to find the best optimum solution.
For the generalized version of the problem, we compared the General Fair Set Cover ({\bf GFSC}) algorithm and the Efficient or Faster GFSC algorithm ({\bf EffGFSC}) along with {\bf Naive} and Standard Greedy ({\bf SC}) algorithms.

{\em Evaluation metrics:}
We used three metrics for evaluating the algorithms: 
(a) {\bf Output Cover Size}, (b) {\bf Fairness Ratio}, and (c) {\bf Running Time}. 
Following our (general) fairness definition in \S~\ref{sec:pre}, 
let the ratio of each group $\gee_h\in\Gee$ in the cover $\cover$ be $R(\gee_h) = \frac{1}{f_h}\big|\{\cover \cap \Sets_h\}\big|$.
Then, the Fairness ratio is computed as $\frac{\min_{\gee_h\in\Gee} R(\gee_h)}{\max_{\gee_h\in\Gee} R(\gee_h)}$.

This formula for count-parity, for example, computes the least common group size divided by the most common group size in the output cover of an algorithm.
A fair algorithm with zero unfairness should return a fairness ratio of 1, while the smaller the fairness ratio, the more unfair the algorithm.

All implementations have been run on a server with 24 core CPU and 128 GB memory with Ubuntu 18 as the OS. The code is publicly available\footnote{\href{https://github.com/UIC-InDeXLab/fair_set_cover}{https://github.com/UIC-InDeXLab/fair\_set\_cover}}.



 \vspace{-3mm}
\subsection{Validation}\label{sec:exp:validation}
\vspace{-1mm}
We begin our experiments by validating the problem we study.
To do so, we used the Resume Skills dataset. Using all of its 1986 rows, 
we sampled different set of skills $N$ varying from 10 to 20 out of all 218 available skills as the ground set.
We run the brute-force algorithms {\sc opt-SC} and {\sc opt-FSC} to find the optimal solutions for the set cover and the fair set cover, respectively.
Next, we used the greedy algorithm for set cover and EffAllPick to find the approximate fair set cover. The results are shown in the following.

\begin{center}
\small
\begin{tabular}{|@{}c@{}|c@{}|c@{}|}
        \hline
        {\bf Algorithm} & {\bf Avg. Fairness Ratio} & {\bf Avg. Cover Size} \\ [0.5ex] 
        \hline
        {\sc Opt-SC} & \textcolor{red}{0.48} & 3.32 \\
        {\sc Greedy-SC} & \textcolor{red}{0.55} & 3.42 \\
        {\sc Opt-FSC} & {\bf 1.00} & 3.75 \\ 
        {\sc EffAllPick} & {\bf 1.00} & {3.90} \\
        \hline
    \end{tabular}
    \label{tab:validation}
\end{center}
First, we observe that both optimal and approximation solutions for the set cover were significantly unfair. This verifies that without considering fairness, set cover can cause major biases.
On the other hand, formulating the problem as the fair set cover, both the optimal and the approximation algorithms have zero unfairness. Second, comparing the average cover size for {\sc Opt-SC} and {\sc Opt-FSC}, one can verify a negligible price of fairness as the avg. cover size increased by only 0.43 (less than half a set). Last but not least, even though {\sc EffAllPick} has a $\log(n)$ approximation ratio, in practice, its cover size is close to the optimal, since its avg. cover size was only 0.15 sets more than the optimal ({\sc Opt-FSC}).

\begin{figure*}[!t]
\centering
    \begin{subfigure}[t]{0.32\linewidth}
        \includegraphics[width=.95\linewidth]{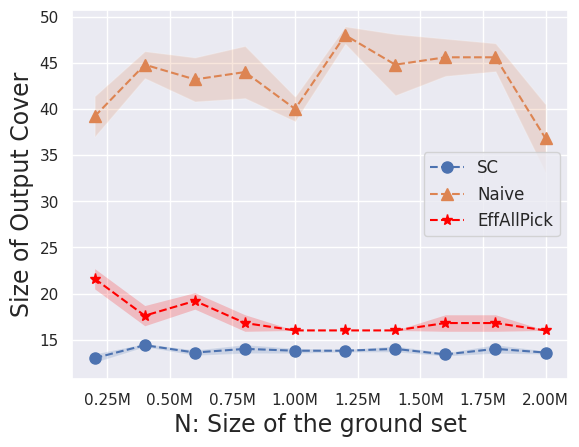}
        \caption{Output size}
    \end{subfigure}
    \hfill
    \begin{subfigure}[t]{0.32\linewidth}
        \includegraphics[width=.95\linewidth]{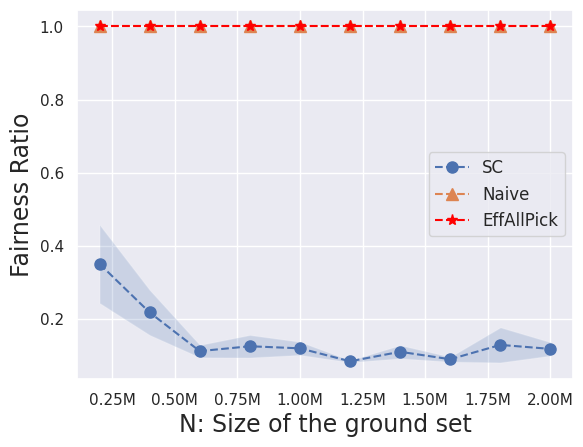}
        \caption{Fairness}
    \end{subfigure}
    \hfill
    \begin{subfigure}[t]{0.32\linewidth}
        \includegraphics[width=.95\linewidth]{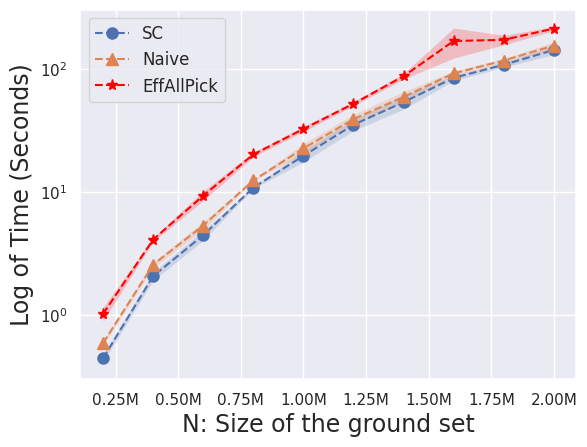}
        \caption{Time}
    \end{subfigure}
 \vspace{-4mm}
\caption{Comparing (a) the size of output cover, (b) fairness based on ratio-parity, and (c) running time (averaged over 5 samples) on POPSIM dataset. The x-axis shows different values of $N$ (size of the ground set).}
\label{fig:exp:popsim:large}
 \vspace{-4mm}
\end{figure*}

\vspace{-3mm}
\subsection{Resume Skills}\label{sec:exp:resume}
In this experiment, we run experiments of team of experts formation, using the Resume Skills dataset on various settings, while using gender to specify the demographic groups {\tt male} and {\tt female}. 
Let $N$ denote the size of the ground set (i.e., the skill sets to cover). For each $N \in \{15, 20, \cdots, 50\}$, we sample $N$ skills 20 times. We ran the experiments on 20 different sample sets (each sample set has a random set of $N$ skills) and averaged the results.
The goal is to find a fair and minimal set of candidates that cover the ground set.
The results, aggregated on all 20 samples, are presented in Fig.~\ref{fig:resume}. The error bounds in all plots show the standard deviation.

\mparagraph{Size of the output cover}
From Fig.~\ref{fig:cover_size_resume}, 
while the output size of the Naive Algorithm is noticeably larger than the greedy set cover (SC), the output sizes for AllPick and EffAllPick algorithms are very close to SC, as the error bars of the three algorithms highly overlap.
Also, it is worth to note that, despite the fact that EffAllPick has a slightly worse approximation ratio than AllPick, in practice, the two algorithms performed near-identical.

\mparagraph{Fairness Ratio} 
As one can see in Fig.~\ref{fig:fairness_resume}, all our algorithms (Naive, AllPick, and EffAllPick) guarantee a fairness ratio of 1, but the standard greedy set cover algorithm returned significantly unfair outputs in all cases, all having a fairness ratio less than 0.8.

\mparagraph{The Running Time}
From Fig.~\ref{fig:running_time_resume}, we can see that the AllPick Algorithm takes a significantly longer time to find a fair cover. 
This is not surprising because its time complexity is exponential to the number of colors.
On the other hand, using the $\mathsf{LP}$ and Randomized Rounding in EffAllPick significantly reduces this running time.
\vspace{-2mm}
\subsection{Fair Survey Collection}\label{sec:exp:real}
In this part, we used two real datasets {\bf Adult} and {\bf COMPAS}. In the {\bf Adult} dataset, the gender is the sensitive attribute. The female-to-male ratio is almost 0.5 in this dataset. We aim to select a subset of individuals (rows) from this dataset to form a set of candidates for a survey that is fair and covers a set of all 30 available criteria. For example, we want to have at least one individual with 'High Income' and at least one with 'Low Income'. Similarly, we want to cover all categories of 'Marital Status' in the final cover. These 30 criteria are our ground-set and the individuals are the family of sets that build the cover. The same procedure is followed for the {\bf COMPAS} dataset. We considered 3 races, "African-American," "Caucasian," and "Hispanic," as the demographic groups (colors) and defined a set of 30 criteria to satisfy in the final cover (i.e., the set of individuals chosen for a survey). We picked a sample of size $M$ from the rows of these datasets. This sample corresponds to a set of individuals covering the ground set (the criteria). In other words, $|\Sets| = M$ where $\Sets$ is the family of subsets defined in the Set Cover problem. We repeated the experiments on different values of $M=\{2000,3000,\ldots, 10000\}$.
The results of these runs are illustrated in Fig ~\ref{fig:exp:adult} and in the technical report~\cite{dehghankar2024fair}. The results are similar to the previous experiment \S~\ref{sec:exp:resume}.
 \vspace{-4mm}
\subsection{POPSIM}\label{sec:exp:popsim}
This dataset contains a set of colored points in $\Re^2$. Our coverage problem on this dataset can be defined as selecting a subset of points (e.g. for business license distribution or survey collection) to cover all the neighbors. Two points are considered neighbors if their distance is less than $r$ for a fixed radius $r$. Based on this definition, each individual is a subset that covers a circle with a radius $r$ of other individuals around it. For example, Fig. 12 in the technical report~\cite{dehghankar2024fair} shows an example of a fair cover and an unfair cover on a sampled subset of this dataset. Each point is an individual, and each circle shows a selected individual and its coverage.

Similar to the previous experiments, we ran our methods on different number of elements (size of ground set). In Fig ~\ref{fig:exp:popsim:large}, we show the output size, fairness and running time of the algorithms over different sizes of the ground set (from 0.25M to 2M). The algorithm AllPick did not finish within $1$ hour so we exclude it from the figure. All observations are the same as in the other datasets: Our EffAillPick algorithm satisfies zero-fairness and it is competitive to SC with respect to the size of the output cover and the running time. The SC algorithm returns highly unfair outputs with fairness ratio less than $0.2$ in most cases. 


 \vspace{-3mm}
\subsection{Extended results}\label{sec:exp:2}
 \vspace{-1mm}
We use our synthetic data 
to verify our methods in different settings like Generalized Fair Set Cover and more demographic groups (colors). 
Due to space limitations, these experiments and the details of our synthetic datasets are provided in the technical report~\cite{dehghankar2024fair}.
\vspace{-3mm}
\section{Related Work}\label{sec:related}
Algorithmic fairness has received a lot of attention in the last few years. 
However, most of the recent work focuses on the intersection of fairness and Machine Learning ~\cite{barocas2023fairness,mehrabi2021survey}. 

Despite its importance, limited work has been done on algorithmic fairness in
the context of combinatorial problems \cite{wang2022balancing}. Satisfying fairness constraints while optimizing sub-modular functions is studied in \cite{wang2022balancing}. Further examples include satisfying fairness in Coverage Maximization problems ~\cite{asudeh2023maximizing,bandyapadhyay2021fair} and Facility Locations~\cite{jung2019center}.
In the red-blue set cover problem~\cite{carr2000red}, the goal is to select a family of sets to cover all red points, while minimizing the number of blue points covered.
In this body of work, the focus is on defining fairness on the point set (ground set). In other words, they assume the points belong to different demographic groups. In contrast, in this work, we define fairness on the family of sets.
Fairness is also studied in problems like Matching \cite{garcia2020fair,esmaeili2023rawlsian}, 
Resource Allocation~\cite{mashiat2022trade}, Ranking~\cite{asudeh2019designing,zehlike2017fa}, Queue Systems~\cite{demers1989analysis}, and Clustering~\cite{makarychev2021approximation,thejaswi2021diversity}. 

There has been recent work studying fairness in the Hitting Set problem~\cite{inamdar2023fixed}. Although hitting set is the dual of the set cover problem, the definition of fairness used is different from ours. 
We define and guarantee perfect fairness based on demographic parity (count and ratio parity). 
In other words, we require an exact equality (e.g., equality of count equality or exact color ratio) between demographic groups.
By contrast, in ~\cite{inamdar2023fixed}, a hitting set is defined to be fair if does not contain {\bf many} points from a color. They have upper bounds on the number of points (equivalent to the subsets in the set cover problem) that can be picked from any particular color.
Furthermore, they design exact (fixed-parameter) algorithms with running time exponential in $k$ and/or exponential on maximum number of sets that an element belongs to.

The Fair Team Formation Problem and the fair set cover problem with two groups are also discussed in ~\cite{barnabo2019algorithms}. The authors form a fair team of experts, covering a specific set of skills and, at the same time, being fair in terms of sensitive attributes. They focus on the binary case where individuals only belong to one of the two demographic groups. They present inapproximability results and some some heuristics. The practical experiments are limited to a dataset of around 6 skills and 1211 individuals.

\section{Conclusion}
Set cover has an extended scope for solving real-world problems with societal impact. In this paper we revisited this problem through the lens of fairness. 
We adopted the group fairness notion of demographic parity, and proposed a general formulation that extends to various cases such as count-parity and ratio-parity for binary and non-binary demographic groups.
We formulated set cover and weighted set cover under fairness constraint and proposed approximation algorithms that (a) always guarantee fairness, (b) have (almost) the same approximation ratio and similar time complexity with the the greedy algorithm for set cover (without fairness).
Our experiments on real and synthetic data across different settings demonstrated the high performance of our algorithms.

\begin{acks}
This work was supported in part by the National Science Foundation, Grant No. 2348919 and 2107290.
\end{acks}

 \newpage
\bibliographystyle{abbrv}
\balance
\bibliography{ref}

\begin{thebibliography}{10}

\bibitem{angwin2016compas}
J.~Angwin, J.~Larson, S.~Mattu, and L.~Kirchner.
\newblock Propublica compas dataset, 2016.

\bibitem{asudeh2023maximizing}
A.~Asudeh, T.~Berger-Wolf, B.~DasGupta, and A.~Sidiropoulos.
\newblock Maximizing coverage while ensuring fairness: A tale of conflicting objectives.
\newblock {\em Algorithmica}, 85(5):1287--1331, 2023.

\bibitem{asudeh2019designing}
A.~Asudeh, H.~Jagadish, J.~Stoyanovich, and G.~Das.
\newblock Designing fair ranking schemes.
\newblock In {\em Proceedings of the 2019 international conference on management of data}, pages 1259--1276, 2019.

\bibitem{bandyapadhyay2021fair}
S.~Bandyapadhyay, A.~Banik, and S.~Bhore.
\newblock On fair covering and hitting problems.
\newblock In {\em Graph-Theoretic Concepts in Computer Science: 47th International Workshop, WG 2021, Warsaw, Poland, June 23--25, 2021, Revised Selected Papers 47}, pages 39--51. Springer, 2021.

\bibitem{barnabo2019algorithms}
G.~Barnab{\`o}, A.~Fazzone, S.~Leonardi, and C.~Schwiegelshohn.
\newblock Algorithms for fair team formation in online labour marketplaces.
\newblock In {\em Companion Proceedings of The 2019 World Wide Web Conference}, pages 484--490, 2019.

\bibitem{barocas2017fairness}
S.~Barocas, M.~Hardt, and A.~Narayanan.
\newblock Fairness in machine learning.
\newblock {\em Nips tutorial}, 1:2017, 2017.

\bibitem{barocas2023fairness}
S.~Barocas, M.~Hardt, and A.~Narayanan.
\newblock {\em Fairness and machine learning: Limitations and opportunities}.
\newblock MIT Press, 2023.

\bibitem{bashardoust2023reducing}
A.~Bashardoust, S.~Friedler, C.~Scheidegger, B.~D. Sullivan, and S.~Venkatasubramanian.
\newblock Reducing access disparities in networks using edge augmentation.
\newblock In {\em Proceedings of the 2023 ACM Conference on Fairness, Accountability, and Transparency}, pages 1635--1651, 2023.

\bibitem{blanco2023fairness}
V.~Blanco and R.~G{\'a}zquez.
\newblock Fairness in maximal covering location problems.
\newblock {\em Computers \& Operations Research}, 157:106287, 2023.

\bibitem{carr2000red}
R.~D. Carr, S.~Doddi, G.~Konjevod, and M.~Marathe.
\newblock On the red-blue set cover problem.
\newblock In {\em Proceedings of the eleventh annual ACM-SIAM symposium on Discrete algorithms}, pages 345--353, 2000.

\bibitem{charikar2001clustering}
M.~Charikar and R.~Panigrahy.
\newblock Clustering to minimize the sum of cluster diameters.
\newblock In {\em Proceedings of the thirty-third annual ACM symposium on Theory of computing}, pages 1--10, 2001.

\bibitem{chen2024parameterized}
X.~Chen, D.~Xu, Y.~Xu, and Y.~Zhang.
\newblock Parameterized approximation algorithms for sum of radii clustering and variants.
\newblock In {\em Proceedings of the AAAI Conference on Artificial Intelligence}, volume~38, pages 20666--20673, 2024.

\bibitem{cho2012chapter}
D.-Y. Cho, Y.-A. Kim, and T.~M. Przytycka.
\newblock Chapter 5: Network biology approach to complex diseases.
\newblock {\em PLoS computational biology}, 8(12):e1002820, 2012.

\bibitem{dehghankar2024fair}
M.~Dehghankar, R.~Raychaudhury, S.~Sintos, and A.~Asudeh.
\newblock Fair set cover.
\newblock {\em arXiv preprint arXiv:2405.11639}, 2024.

\bibitem{demers1989analysis}
A.~Demers, S.~Keshav, and S.~Shenker.
\newblock Analysis and simulation of a fair queueing algorithm.
\newblock {\em ACM SIGCOMM Computer Communication Review}, 19(4):1--12, 1989.

\bibitem{doddi2000approximation}
S.~R. Doddi, M.~V. Marathe, S.~S. Ravi, D.~S. Taylor, and P.~Widmayer.
\newblock Approximation algorithms for clustering to minimize the sum of diameters.
\newblock In {\em Scandinavian Workshop on Algorithm Theory}, pages 237--250. Springer, 2000.

\bibitem{donahue2020fairness}
K.~Donahue and J.~Kleinberg.
\newblock Fairness and utilization in allocating resources with uncertain demand.
\newblock In {\em Proceedings of the 2020 conference on fairness, accountability, and transparency}, pages 658--668, 2020.

\bibitem{Dua2019}
D.~Dua and C.~Graff.
\newblock Uci machine learning repository: Adult data set, 2019.
\newblock Accessed: 2024-05-20.

\bibitem{esmaeili2023rawlsian}
S.~Esmaeili, S.~Duppala, D.~Cheng, V.~Nanda, A.~Srinivasan, and J.~P. Dickerson.
\newblock Rawlsian fairness in online bipartite matching: Two-sided, group, and individual.
\newblock In {\em Proceedings of the AAAI Conference on Artificial Intelligence}, volume~37, pages 5624--5632, 2023.

\bibitem{garcia2020fair}
D.~Garc{\'\i}a-Soriano and F.~Bonchi.
\newblock Fair-by-design matching.
\newblock {\em Data Mining and Knowledge Discovery}, 34:1291--1335, 2020.

\bibitem{ge2010application}
X.~Ge.
\newblock An application of covering approximation spaces on network security.
\newblock {\em Computers \& Mathematics with Applications}, 60(5):1191--1199, 2010.

\bibitem{hansen1997cluster}
P.~Hansen and B.~Jaumard.
\newblock Cluster analysis and mathematical programming.
\newblock {\em Mathematical programming}, 79(1-3):191--215, 1997.

\bibitem{hirschman2007black}
J.~Hirschman, S.~Whitman, and D.~Ansell.
\newblock The black: white disparity in breast cancer mortality: the example of chicago.
\newblock {\em Cancer Causes \& Control}, 18:323--333, 2007.

\bibitem{hotegni2023approximation}
S.~S. Hotegni, S.~Mahabadi, and A.~Vakilian.
\newblock Approximation algorithms for fair range clustering.
\newblock In {\em International Conference on Machine Learning}, pages 13270--13284. PMLR, 2023.

\bibitem{inamdar2023fixed}
T.~Inamdar, L.~Kanesh, M.~Kundu, N.~Purohit, and S.~Saurabh.
\newblock Fixed-parameter algorithms for fair hitting set problems.
\newblock In {\em 48th International Symposium on Mathematical Foundations of Computer Science}, 2023.

\bibitem{inamdar2020capacitated}
T.~Inamdar and K.~Varadarajan.
\newblock Capacitated sum-of-radii clustering: An fpt approximation.
\newblock In {\em 28th Annual European Symposium on Algorithms (ESA 2020)}, 2020.

\bibitem{jiang2021rawlsian}
Y.~Jiang, X.~Wu, B.~Chen, and Q.~Hu.
\newblock Rawlsian fairness in push and pull supply chains.
\newblock {\em European Journal of Operational Research}, 291(1):194--205, 2021.

\bibitem{jones2020fair}
M.~Jones, H.~Nguyen, and T.~Nguyen.
\newblock Fair k-centers via maximum matching.
\newblock In {\em International conference on machine learning}, pages 4940--4949. PMLR, 2020.

\bibitem{jung2019center}
C.~Jung, S.~Kannan, and N.~Lutz.
\newblock A center in your neighborhood: Fairness in facility location.
\newblock {\em arXiv preprint arXiv:1908.09041}, 2019.

\bibitem{kleinberg2016inherent}
J.~Kleinberg, S.~Mullainathan, and M.~Raghavan.
\newblock Inherent trade-offs in the fair determination of risk scores.
\newblock In {\em 8th Innovations in Theoretical Computer Science Conference (ITCS 2017)}, 2017.

\bibitem{kleindessner2019fair}
M.~Kleindessner, P.~Awasthi, and J.~Morgenstern.
\newblock Fair k-center clustering for data summarization.
\newblock In {\em International Conference on Machine Learning}, pages 3448--3457. PMLR, 2019.

\bibitem{ma2021fair}
Y.~Ma and J.~Zheng.
\newblock Fair regret minimization queries.
\newblock In {\em International Conference on Intelligent Data Engineering and Automated Learning}, pages 511--523. Springer, 2021.

\bibitem{makarychev2021approximation}
Y.~Makarychev and A.~Vakilian.
\newblock Approximation algorithms for socially fair clustering.
\newblock In {\em Conference on Learning Theory}, pages 3246--3264. PMLR, 2021.

\bibitem{adBiasInfluence}
C.~Marcelle.
\newblock Analysis of influencer marketing and social media diversity reflects visual bias.
\newblock \url{chantellemarcelle.com/influencer-marketing-reflects-social-media-diversity-issue/}.

\bibitem{mashiat2022trade}
T.~Mashiat, X.~Gitiaux, H.~Rangwala, P.~Fowler, and S.~Das.
\newblock Trade-offs between group fairness metrics in societal resource allocation.
\newblock In {\em Proceedings of the 2022 ACM Conference on Fairness, Accountability, and Transparency}, pages 1095--1105, 2022.

\bibitem{mehrabi2021survey}
N.~Mehrabi, F.~Morstatter, N.~Saxena, K.~Lerman, and A.~Galstyan.
\newblock A survey on bias and fairness in machine learning.
\newblock {\em ACM computing surveys (CSUR)}, 54(6):1--35, 2021.

\bibitem{monma1989partitioning}
C.~Monma and S.~Suri.
\newblock Partitioning points and graphs to minimize the maximum or the sum of diameters.
\newblock In {\em Graph Theory, Combinatorics and Applications (Proc. 6th Internat. Conf. Theory Appl. Graphs)}, volume~2, pages 899--912, 1989.

\bibitem{nguyen2023popsim}
K.~D. Nguyen, N.~Shahbazi, and A.~Asudeh.
\newblock Popsim: An individual-level population simulator for equitable allocation of city resources.
\newblock Algorithmic Fairness in Artificial intelligence, Machine learning and and Decision making (AFair-AMLD23), 2023.

\bibitem{pessach2022review}
D.~Pessach and E.~Shmueli.
\newblock A review on fairness in machine learning.
\newblock {\em ACM Computing Surveys (CSUR)}, 55(3):1--44, 2022.

\bibitem{revelle1976applications}
C.~ReVelle, C.~Toregas, and L.~Falkson.
\newblock Applications of the location set-covering problem.
\newblock {\em Geographical analysis}, 8(1):65--76, 1976.

\bibitem{rubin1973technique}
J.~Rubin.
\newblock A technique for the solution of massive set covering problems, with application to airline crew scheduling.
\newblock {\em Transportation Science}, 7(1):34--48, 1973.

\bibitem{hp1}
M.~Simon.
\newblock {HP} looking into claim webcams can't see black people.
\newblock CNN, 2009.

\bibitem{swift2022maximizing}
I.~P. Swift, S.~Ebrahimi, A.~Nova, and A.~Asudeh.
\newblock Maximizing fair content spread via edge suggestion in social networks.
\newblock {\em Proceedings of the VLDB Endowment}, 15(11):2692--2705, 2022.

\bibitem{thejaswi2021diversity}
S.~Thejaswi, B.~Ordozgoiti, and A.~Gionis.
\newblock Diversity-aware k-median: Clustering with fair center representation.
\newblock In {\em Machine Learning and Knowledge Discovery in Databases. Research Track: European Conference, ECML PKDD 2021, Bilbao, Spain, September 13--17, 2021, Proceedings, Part II 21}, pages 765--780. Springer, 2021.

\bibitem{hp2}
T.~Townsend.
\newblock Most engineers are white and so are the faces they use to train software.
\newblock Recode, 2017.

\bibitem{tsang2019group}
A.~Tsang, B.~Wilder, E.~Rice, M.~Tambe, and Y.~Zick.
\newblock Group-fairness in influence maximization.
\newblock In {\em Proceedings of the Twenty-Eighth International Joint Conference on Artificial Intelligence, {IJCAI}}, pages 5997--6005. ijcai.org, 2019.

\bibitem{vemuganti1998applications}
R.~R. Vemuganti.
\newblock Applications of set covering, set packing and set partitioning models: A survey.
\newblock {\em Handbook of Combinatorial Optimization: Volume1--3}, pages 573--746, 1998.

\bibitem{chiagoDispensry}
A.~Vinicky.
\newblock While a black-owned cannabis dispensary opens in chicago, critics say state’s equity work still falling short.
\newblock WTTW, 2022.

\bibitem{wang2022balancing}
Y.~Wang, Y.~Li, F.~Bonchi, and Y.~Wang.
\newblock Balancing utility and fairness in submodular maximization (technical report).
\newblock {\em arXiv preprint arXiv:2211.00980}, 2022.

\bibitem{xu2005survey}
R.~Xu and D.~Wunsch.
\newblock Survey of clustering algorithms.
\newblock {\em IEEE Transactions on neural networks}, 16(3):645--678, 2005.

\bibitem{young2008greedy}
N.~E. Young.
\newblock Greedy set-cover algorithms (1974-1979, chv{\'a}tal, johnson, lov{\'a}sz, stein).
\newblock {\em Encyclopedia of algorithms}, pages 379--381, 2008.

\bibitem{zehlike2017fa}
M.~Zehlike, F.~Bonchi, C.~Castillo, S.~Hajian, M.~Megahed, and R.~Baeza-Yates.
\newblock Fa* ir: A fair top-k ranking algorithm.
\newblock In {\em Proceedings of the 2017 ACM on Conference on Information and Knowledge Management}, pages 1569--1578, 2017.

\end{thebibliography}

\appendix

\vspace{-2mm}
\section{Missing algorithms and proofs from \S~\ref{sec:WFSC}}
\label{appndx:sec5}
\vspace{-1mm}
\subsection{Naive Algorithm}
We first design a naive algorithm for the \fairWprob{}, similar to the naive algorithm we proposed for the \fairprob{}.
We execute the well-known greedy algorithm for the weighted set cover problem in the instance $(\Elements, \Sets)$ without considering the colors of the sets.
Let $C$ be the family of sets returned by the greedy algorithm. Then we add arbitrary sets from each color to equalize the number of sets from each color in the final cover. Let $\Wcover$ be the final set we return.

If $\alpha$ is the approximation factor of the greedy algorithm in the weighted set cover, we show that the naive algorithm returns an $\alpha  k\Wratio$-approximation. Let $C^*$ be the optimum solution for the weighted set instance $(\Elements, \Sets)$ and $\optWCover$ be the optimum solution for the \fairWprob{} in the instance $(\Elements, \Sets, \weight)$. Notice that $\weight(\optWCover)\geq \weight(C^*)$.
In the worst case, $C$ contains all the sets from one of the colors $\group_t$. As a result, for all other colors, we should add $|C|$ arbitrary sets. For each $h \leq k$ and $h \neq t$, let $A_h$ be the family of sets we added to $C$ to form $\Wcover$. Based on the definition, $|A_h| = k$ for every $h=1,\ldots, k$, and $\Wcover = C \cup (\bigcup_{h \leq k, h \neq t} A_j)$. We have:

\vspace{-3mm}


\begin{align*}
    \frac{\weight(\Wcover)}{\weight(\optWCover)} &= \frac{\weight(C \cup (\bigcup_{h \leq k, h \neq t} A_j))}{\weight(\optWCover)} \\
    &= \frac{\weight(C)}{\weight(\optWCover)} + \sum_{h \leq k, h \neq t} \frac{\weight(A_j)}{\weight(\optWCover)}\leq \frac{\weight(C)}{\weight(C^*)} + \sum_{h \leq k, h \neq t} \frac{\weight(A_j)}{\weight(C^*)}\\
    &\leq\!\alpha\!+\! \alpha\!\!\!\!\!\!\!\sum_{h \leq k, h \neq t}\!\!\!\frac{\weight(A_j)}{\weight(C^)}\leq \alpha (1 +\!\!\!\!\!\!\sum_{h \leq k, h \neq t} \!\!\!\!\!\!\Wratio)= \alpha (1 + (k - 1) \Wratio ) = k \Wratio \alpha
\end{align*}

We know that $\alpha=\ln n +1$, so we conclude with Theorem~\ref{thm:NaivefairWCover}.

\vspace{-3mm}
\subsection{Proof of Theorem~\ref{thm:fairWCover}}
\vspace{-1mm}
First, it is straightforward to see that $\Wcover$ is a fair cover. The algorithm stops when there is no uncovered element so $\Wcover$ is a cover. In every iteration, we add exactly one set of color $\group_1$ and one set of color $\group_2$, so $\Wcover$ is a fair cover.

Next, we show that our algorithm computes an $\Wratio(\ln n +1)$-approximation solution for the \fairWprob{} problem.
\begin{lemma}
    \label{lem:fairWapprox}
    $\weight(\Wcover)\leq \Wratio(\ln n +1)\cdot \weight(\optWCover)$.
\end{lemma}
\begin{proof}
    If $\set_A, \set_B\in \Sets$, let $\set_{A,B}=\set_A\cup \set_B$.
We define the set $\Elements'=\Elements$ and $\Sets'=\{\set_{i,j}\mid \set_i\in \Sets_1, \set_j\in \Sets_2\}$. We also define the weighted function $\weight':\Sets'\rightarrow \Re^+$ such that $\weight'(\set_{A,B})=\weight(\set_A)+\weight(\set_B)$. Let $C^*$ be the optimum solution of the weighted set cover instance $(\Elements', \Sets', \weight')$. Any fair cover returned by our algorithm can be straightforwardly mapped to a valid weighted set cover for $(\Elements', \Sets')$. For example, if our algorithm selects a pair $\set_A, \set_B$ then it always holds that $\set_{A,B}\in \Sets'$, $|(\set_A\cap \set_B)\cap \Elements|=|\set_{A,B}\cap \Elements'|$, and $\weight'(\set_{A,B})=\weight(\set_A)+\weight(\set_B)$.
We have $\weight'(C^*)\leq \weight(\optCover)$.

Recall that the standard greedy algorithm for the weighted set cover problem returns an $(\ln n +1)$-approximation.
We show that our algorithm implements a variation of such a greedy algorithm and returns a $\Wratio(\ln n +1)$-approximation in the weighted set cover instance $(\Elements', \Sets', \weight')$. In particular, we show that in any iteration, our algorithm chooses a pair of sets $\set_A, \set_B$ such that the set $\frac{\weight'(\set_{A,B})}{|\set_{A,B}\cap \uncovered|}$ is a $\Wratio$-approximation of the best ratio among the available sets.

At the beginning of an iteration $i$, assume that our algorithm has selected the pairs $(\set_{A_1}, \set_{B_1}),\ldots, (\set_{A_{i-1}}, \set_{B_{i-1}})$, so the sets 
$\set_{A_1,B_1}, \ldots,\\ \set_{A_{i-1}, B_{i-1}}$ have been selected for the set cover instance $(\Elements', \Sets')$. Let $\uncovered$ be the set of uncovered elements in $\Elements'$ at the beginning of the $i$-th iteration.
Let $\set_{A_j,B_j}$ be the set in $\Sets'$ (at the beginning of the $i$-th iteration) with the minimum $\frac{\weight'(\set_{A_j, B_j})}{|\set_{A_j,B_j}\cap \uncovered|}$. We consider two cases.

If $\set_{A_j}\notin \{\set_{A_1},\ldots, \set_{A_{i-1}}\}$ and $\set_{B_j}\notin \{\set_{B_1},\ldots, \set_{B_{i-1}}\}$, then by definition our greedy algorithm selects the pair $(\set_{A_j}, \set_{B_j})$ so the set $\set_{A_j,B_j}$ is added in the cover $(\Elements', \Sets')$.

Next, consider the second case where $\set_{A_j}\notin \{\set_{A_1},\ldots, \set_{A_{i-1}}\}$ and $\set_{B_j}\in \{\set_{B_1},\ldots, \set_{B_{i-1}}\}$. Since there are enough sets from each group, let $S_{B_h}$ be a set with color $\group_2$ such that $\set_{B_h}\notin \{\set_{B_1},\ldots, \set_{B_{i-1}}\}$.
By definition, we have $|(\set_{A_j}\cup \set_{B_h})\cap \uncovered|\geq |(\set_{A_j}\cup \set_{B_j})\cap \uncovered|$.
Hence, it holds that $\weight(\set_{B_h})\geq \weight(\set_{B_j})$, otherwise we would have $\frac{\weight'(\set_{A_j,B_h})}{|\set_{A_j,B_h}\cap \uncovered|}< \frac{\weight'(\set_{A_j,B_j})}{|\set_{A_j,B_j}\cap \uncovered|}$, which is contradiction.
By definition it holds that $\Wratio\geq \frac{\weight(s_{B_h})}{\weight(s_{B_j})}$.
Our algorithm selects the pair such that $\frac{\weight(\set_{A^*})+\weight(\set_{B^*})}{|(\set_{A^*}\cup\set_{B^*})\cap \uncovered|}$ is minimized.
Putting everything together, we have, 
\begin{align*}
&\frac{\weight'(\set_{A^*,B^*})}{|\set_{A^*, B^*}\cap \uncovered|}=\frac{\weight(\set_{A^*})+\weight(\set_{B^*})}{|(\set_{A^*}\cup\set_{B^*})\cap \uncovered|}\leq \frac{\weight(\set_{A_j})+\weight(\set_{B_h})}{|(\set_{A_j}\cup \set_{B_h})\cap \uncovered|}\\&\leq\!\!\! \frac{\weight(\set_{A_j})+\weight(\set_{B_h})}{|(\set_{A_j}\cup \set_{B_j})\cap \uncovered|}\!\!\leq\!\! \frac{\weight(\set_{A_j})+\Wratio\cdot\weight(\set_{B_j})}{|(\set_{A_j}\cup \set_{B_j})\cap \uncovered|}\!\!\leq\!\! \Wratio\frac{\weight(\set_{A_j})+\weight(\set_{B_j})}{|(\set_{A_j}\cup \set_{B_j})\cap \uncovered|}.
\end{align*}
It is known~\cite{young2008greedy} that if a greedy algorithm for the weighted set cover computes a $\beta$-approximation of the best ratio in each iteration, then the algorithm returns a $\beta(\ln n +1)$-approximation.
Hence, $\weight(\Wcover)\leq \Wratio(\ln n +1)\weight'(C^*)\leq \Wratio(\ln n+1)\weight(\optWCover)\Leftrightarrow \weight(\Wcover)\leq \Wratio(\ln n +1)\weight(\optWCover)$.
\end{proof}

At the beginning of the algorithm, for every pair of sets we compute the number of elements they cover in $O(m^2n)$ time. Each time we cover a new element we update the counters in $O(m^2)$ time. In total, our algorithm runs in $O(m^2n)$ time.

Our algorithm can be extended to $k$ groups, straightforwardly. In each iteration, we find the $k$ sets $\set_{A_1}\in \Sets_1, \ldots, \set_{A_k}\in \Sets_k$ such that the ratio $\frac{\sum_{1\leq i\leq k}\weight(\set_{A_i})}{|\left(\bigcup_{1\leq i\leq k}S_{A_i}\right)\cap \uncovered|}$ is minimized. The running time increases to $O(m^kn)$.

\vspace{-3mm}
\subsection{Proof of Theorem~\ref{thm:WeightfairMaxKCover}}
\vspace{-1mm}
First, by definition, $\kcover^{(\tau^*)}$ contains exactly one set from every group $\group_i$. Hence $\kcover^{(\tau^*)}$ always satisfies the fairness requirement. Next, we show the approximation factor and the running time.

We introduce some useful notation. Assume that the optimum solution for the \wkmc instance $(\Elements^-, \Sets^-, \weight)$ has value $\mathsf{OPT}=\frac{\mathsf{OPT_\weight}}{\tau'}$ be the optimum solution covers $\tau'$ elements in $\Elements^-$ and the sum of weights of all sets in the optimum solution is $\mathsf{OPT}_\weight$. For any $\tau\in[1,n]$, let $\mathsf{OPT}_{\mathsf{LP}}^{(\tau)}=\sum_{\set_i\in\Sets^-}\weight(\set_i)x_i^{(\tau)}$, i.e., the optimum solution of the $\mathsf{LP}$ with parameter $\tau$. Finally, for $\tau=1,\ldots, n$, let $Z_i^{(\tau)}$ be the random variable which is $1$ if $\set_i$ is selected in the cover, otherwise it is $0$.
\vspace{-2mm}
\begin{lemma}
    The expected weight of the family of sets $\kcover^{(\tau)}$ is $\mathsf{OPT}_{\mathsf{LP}}^{(\tau)}$.
\end{lemma}
\begin{proof}
    The expected weight of $\kcover$ is $E[\sum_{i}Z_i^{(\tau)}\weight(\set_i)]=\sum_i \weight(\set_i)E[Z_i^{(\tau)}]=\sum_{\group_h\in \Groups}\sum_{\set_i\in \Sets^-_h}\weight(\set_i)E[Z_i^{(\tau)}]=\sum_{\group_h\in \Groups}\sum_{\set_i\in \Sets^-_h}\weight(\set_i)x_i^{(\tau)}=\mathsf{OPT}_{\mathsf{LP}}^{(\tau)}$.
\end{proof}
\vspace{-2mm}
\begin{lemma}
    Our randomized algorithm returns an $\frac{e}{e-1}$-approximation for the \wkmc problem. The approximation factor holds in expectation.
\end{lemma}
\begin{proof}
\vspace{-2mm}    $\frac{E[\sum_{i}Z_i^{(\tau^*)}\weight(\set_i)]}{|\kcover^{(\tau^*)}\cap \Elements^-|}=\frac{\mathsf{OPT}_{\mathsf{LP}}^{(\tau^*)}}{|\kcover^{(\tau^*)}\cap \Elements^-|}\leq \frac{\mathsf{OPT}_{\mathsf{LP}}^{(\tau')}}{|\kcover^{(\tau')}\cap \Elements^-|}\leq \frac{\mathsf{OPT}_{\mathsf{LP}}^{(\tau')}}{(1-\frac{1}{e})\tau'}\leq \frac{\mathsf{OPT}_\weight}{(1-\frac{1}{e})\tau'}=\frac{e}{e-1}\mathsf{OPT}$.
\end{proof}

Next, we focus on the running time analysis. In total we need to solve $n$ linear programs with $n+k$ variables and $2n+2\mu+1$ constraints, so the running time to execute all the linear programs is $O(n\mathcal{L}(n+k,2n+2\mu+1))$.
For each $\kcover^{(\tau)}$ we sample, we compute $|\kcover^{(\tau)}\cap \Elements^-|$ in $O(kn)$ time.

It remains to bound the number of times we sample for every $\tau$ to get $|\kcover^{(\tau)}\cap \Elements^-|\geq \frac{1}{2}(1-1/e)\tau$.
Let $M_j^{(\tau)}$ be a random variable which is $1$ if $\element_j$ is covered, and $0$ otherwise, if sampling according to $x_i^{(\tau)}$ is applied.
\vspace{-2mm}
\begin{lemma}
    \label{lem:technPr}
    $Pr[\sum_{j}M_j^{(\tau)}\geq \frac{1}{2}(1-\frac{1}{e})\tau]\geq \frac{1-1/e}{2n}$.
\end{lemma}
\begin{proof}
\vspace{-2mm}
    We first compute the expected value $E[\sum_{j}M_j^{(\tau)}]$.
    As we had in the proof of Lemma~\ref{lem:lem1}, $Pr[\element_j \text{ is covered after sampling }]\geq (1-1/e)y_j^{(\tau)}$. We have $E[\sum_{\element_j\in \Elements^-}M_j^{(\tau)}]=\sum_{\element_j\in \Elements^-}E[M_j^{(\tau)}]\geq \sum_{\element_j\in \Elements^-}(1-1/e)y_j^{(\tau)}\geq (1-\frac{1}{e})\tau$.

    Let $M^{(\tau)}=\sum_{j}M_j^{(\tau)}$. Using the reverse Markov inequality (for a random variable $V$ such that $Pr[V\leq a]=1$ for a number $a\in\Re$, then for $b<E[V]$, it holds $Pr[V>b]\geq \frac{E[V]-b}{a-b}$), we have $Pr[M^{(\tau)}>\frac{1}{2}(1-1/e)\tau]\geq \frac{E[M^{(\tau)}]-\frac{1}{2}(1-1/e)\tau}{n-\frac{1}{2}(1-1/e)\tau}\\\geq \frac{(1-1/e)\tau-\frac{1}{2}(1-1/e)\tau}{n-\frac{1}{2}(1-1/e)\tau}=\frac{\frac{1}{2}(1-1/e)\tau}{n-\frac{1}{2}(1-1/e)\tau}\geq \frac{1-1/e}{2n}$.
\end{proof}

Notice that $Pr[\sum_{j}M_j^{(\tau)}\geq \frac{1}{2}(1-\frac{1}{e})\tau]$ can be seen as a probability of success in a geometric distribution. In expectation we need $\frac{1}{Pr[\sum_{j}M_j^{(\tau)}\geq \frac{1}{2}(1-\frac{1}{e})\tau]}=O(n)$ trials to get a family of sets $\kcover^{(\tau)}$ that cover at least $\frac{1}{2}(1-\frac{1}{e})\tau$ elements in $\Elements^-$. Hence, we need to repeat the sampling procedure $O(n)$ times in expectation. Each time that we sample a family of $k$ sets $\kcover^{(\tau)}$, we spend $O(kn)$ time to compute $|\kcover^{(\tau)}\cap \Elements^-|$, so for each $\tau$ we spend $O(kn^2)$ expected time in addition to solving the \textsf{LP}. In total, our algorithm runs in $O(n\mathcal{L}(n+k,2n+2\mu+1)+kn^3)$ expected time.

\vspace{-2mm}
\section{Missing algorithms and proofs from \S~\ref{sec:general}}
\label{appndx:Sec6}
\vspace{-1mm}

As a pre-processing step, our algorithm computes $\newsizeGroup_h$ for each $\group_h$ as follows. Recall that each $\fractionGroup_h$ is rational. Without loss of generality, we can assume that each $\fractionGroup_h$ is represented in its simplest form. 
Let $p$ represent the least common multiple (LCM) of the denominators within the set of fractions ${\fractionGroup_1, \ldots, \fractionGroup_h}$. Using standard techniques, $p$ can be computed efficiently. Set $p_h=\fractionGroup_h\cdot p$ for each group $g_h\in\Groups$. The performance of our algorithm will have a dependence on $p$. Hence, it is important to observe that $p$ cannot exceed $|\Sets|=\mu$. This is because in order to satisfy the given fractional constraints, the size of any cover must be at least $p$. However, if $p>\mu$, then this implies the given instance cannot admit such a fractional cover.

In this section, we need the assumption that every group $\group_h\in\Groups$ contain at least $\Omega(\log n)\fractionGroup_h|\optCover|$ sets. For simplicity, we also assume that every group contains the same number of sets, denoted by $m$.
 \vspace{-3mm}
\subsection{Greedy Algorithm}
\vspace{-1mm}
The greedy algorithm for the generalized \fairprob{} is similar to the greedy algorithm from \S~\ref{subsec:greedy}. Instead of choosing one set from each color, in each iteration, we choose a family $C$ of $\newsize$ sets such that i) $C$ contains $\newsizeGroup_h$ sets from group $\group_h$, and ii) $C$ covers the most uncovered elements in $\Elements$.
The proof of Lemma~\ref{lem:fairapprox} can be applied straightforwardly to the generalized \fairprob{}, getting the same approximation guarantee. In each iteration we visit $O(\prod_{h}m^{\newsizeGroup_h})=O(m^{\newsize})$ families of $\newsize$ sets. We conclude to the next theorem.
\vspace{-1mm}
\begin{theorem}
    There exists an $(\ln n + 1)$-approximation algorithm for the Generalized \fairprob{} problem that runs in $O(m^{\newsize} n)$ time.
\end{theorem}
\vspace{-3mm}
\subsection{Faster Algorithm}
\vspace{-1mm}
In order to make the above algorithm faster, we use a generalized version of the randomized rounding introduced for the \fairprob{}. We define the next Integer Program. The only difference with the IP in \S~\ref{subsec:fasterAlg} is that for every group $\group_h$, we should satisfy $\sum_{i:\set_i\in\Sets^-_h}x_i=\newsizeGroup_h$ instead of $\sum_{i:\set_i\in\Sets^-_h}x_i=1$.
\vspace{-1mm}
\begin{align}
\max \sum_j y_j&\\
\sum_{i: \set_i\in\Sets^-_h}x_i&=\newsizeGroup_h,\quad \forall \group_h\in \Groups\\
\sum_{i: \element_j\in \set_i} x_i&\geq y_j, \quad \forall \element_j\in \Elements^-\\
x_i&\in \{0,1\},\quad\forall \set_i\in\Sets^-\\
y_j&\in \{0,1\},\quad\forall \element_j\in \Elements^-
\end{align}
We then replace the integer variables to continuous, as following, and use a polynomial time \textsf{LP} solver to compute the solution of the \textsf{LP} relaxation.

\vspace{-9mm}
\begin{align}
x_i&\in [0,1],\quad\forall \set_i\in\Sets^-\\
y_j&\in [0,1],\quad\forall \element_j\in \Elements^-
\end{align}

Let $x_i^*$ and $y_j^*$ be the values of variables in the optimum solution.
We normalize the values $x_i^*$ defining $z_i^* = x^*_i/\newsizeGroup_h$ for every color $\group_h\in \Groups$, and every $\set_i\in\Sets^-_h$.
For each color $\group_h$, we sample uniformly at random $\newsizeGroup_h$ sets with replacement according to the probabilities $z^*_i$. Let $\kcover$ be the family of sets we sample.

If for a color $\group_h$, the number of distinct sets picked from the sampling procedure is less than $\newsizeGroup_h$, then we add arbitrary sets from this color in order to have exactly $\newsizeGroup_h$ sets.
Similarly to the proof of Lemma~\ref{lem:lem1}, we have,
\vspace{-5mm}
\begin{align*}
Pr[\element_j &\text{ not covered by a set in }\kcover]\leq\prod_{\group_h\in\Groups}\left(1-\sum_{i: \element_j\in \set_i, \set_i\in\Sets^-_h}\frac{x_i^*}{\newsizeGroup_h}\right)^{\newsizeGroup_h}\\&
\leq \prod_{\group_h\in\Groups}\left(\frac{1}{e^{\sum_{i: \element_j\in \set_i, \set_i\in\Sets^-_h}\frac{x_i^*}{\newsizeGroup_h}}}\right)^{\newsizeGroup_h}
=\prod_{\group_h\in\Groups} \frac{1}{e^{\sum_{i: \element_j\in \set_i, \set_i\in\Sets^-_h}x_i^*}}\\&
=\frac{1}{e^{\sum_{\group_h\in\Groups}\sum_{i: \element_j\in \set_i, \set_i\in\Sets^-_h}x_i^*}}
=\frac{1}{e^{\sum_{i: \element_j\in \set_i}x_i^*}} 
\leq \frac{1}{e^{y_j^*}}.
\end{align*}
Let $Z_i$ be a random variable that returns $1$ if $\element_j\in \kcover$, and $0$ otherwise. We have,
$E[\sum_j Z_j] \geq (1 - \frac{1}{e}) OPT_{LP} \geq (1 - \frac{1}{e}) OPT$.

After solving the \textsf{LP} in $\mathcal{L}(n + k, 2n + 2m \cdot k)$ time, we sample $\newsizeGroup_h$ times from each color $\group_h$. Overall, we spend $O(\mathcal{L}(n + k \cdot 2n + 2m \cdot k) + \newsize)$ time.
The randomized rounding algorithm is executed in every iteration of the greedy algorithm.
The greedy algorithm executes $O(m)$ iterations. Using an analysis similar to \S~\ref{subsec:fasterAlg}, we bound the expected number of iterations to $O(\frac{|\optCover|}{\newsize}\log n)$.
Following the same arguments as in \S~\ref{subsec:fasterAlg}, we conclude with the next theorem.
\vspace{-1mm}
\begin{theorem}
    There is a randomized $\frac{e}{e - 1} (\ln n + 1)$-approximation algorithm for the Generalized \fairprob{} problem that runs in $O(m(\newsize + \mathcal{L}(n + k, 2n + 2m \cdot k)))$ time. The approximation factor holds in expectation. The same algorithm also runs in $O(\frac{|\optCover|}{\newsize}\mathcal{L}(n + k, 2n + 2m \cdot k)\log n)$ expected time.
\end{theorem}
\vspace{-4mm}
\subsection{Extensions to  \GfairWprob{}}
\vspace{-1mm}
Using the results in \S~\ref{sec:WFSC}, all our algorithms for the weighted case can be extended to the generalized \fairprob{}. Skipping the details, we have the following theorem.
\vspace{-1mm}
\begin{theorem}
    \label{thm:fairGWCover}
    There exists a $\Wratio(\ln n + 1)$-approximation algorithm for the Generalized \fairWprob{} problem that runs in $O(m^\newsize n)$ time.

Furthermore, there exists a randomized $\frac{e}{e-1}\Wratio(\ln n + 1)$-approximation algorithm for the Generalized \fairWprob{} problem that runs in $O(mn\cdot \mathcal{L}(n+k,2n+2m\cdot k+1)+m\newsize n^3)$ time. The approximation factor and the running time hold in expectation.
\end{theorem}

\end{document}